\documentclass[
 10pt, aps, prx, twocolumn,
 amsfonts, amsmath, amssymb, 
 superscriptaddress, nofootinbib,
 floatfix
]{revtex4-2}

\usepackage{xcolor}
\usepackage{amsthm}
\usepackage{thmtools}
\usepackage{thm-restate}
\usepackage{braket}
\usepackage{graphicx}
\usepackage{booktabs}
\usepackage{mathtools}
\usepackage[ruled,vlined,linesnumbered]{algorithm2e}
\usepackage{microtype}
\usepackage[hypertexnames=false]{hyperref}
\usepackage[capitalise,nameinlink]{cleveref}

\IncMargin{1em}

\declaretheorem{theorem}

\declaretheorem[sibling=theorem]{lemma}

\declaretheorem[style=definition]{definition}

\definecolor{red4}{Hsb}{0,0.79,0.72}
\definecolor{blue4}{Hsb}{240,0.65,0.85}
\definecolor{purple4}{Hsb}{330,0.79,0.70}

\crefname{section}{Sec.}{Secs.}

\DeclareMathOperator*{\argmin}{arg\,min}

\let\oldnl\nl
\newcommand{\nlnonumber}{\renewcommand{\nl}{\let\nl\oldnl}}

\makeatletter
\newcommand{\apptocfile}{atoc}
\let\apptoc@orig@appendix\appendix
\renewcommand{\appendix}{%
  \apptoc@orig@appendix
  \let\apptoc@orig@addtocontents\addtocontents
  \long\def\addtocontents##1##2{%
    \def\apptoc@ext{##1}%
    \def\apptoc@toc{toc}%
    \ifx\apptoc@ext\apptoc@toc
      \apptoc@orig@addtocontents{\apptocfile}{##2}%
    \else
      \apptoc@orig@addtocontents{##1}{##2}%
    \fi
  }%
}
\newcommand{\appendixtableofcontents}{%
  \begingroup
    \setcounter{tocdepth}{3}%
    \phantomsection
    \let\addcontentsline\@gobblethree
    \section*{Contents}%
    \pdfbookmark[1]{Appendices}{apxcontents}%
    \@starttoc{\apptocfile}%
  \endgroup
}
\renewcommand\paragraph{\@startsection{paragraph}{4}{\z@}%
  {1ex \@plus1ex \@minus.2ex}%
  {-1em}%
  {\normalfont\normalsize\bfseries}%
}%
\newcommand{\titlename}{\textsl{\@title}}
\makeatother

\newcommand{\ketbra}[2]{\vert{#1}\rangle\langle{#2}\vert}

\renewcommand{\vec}[1]{\boldsymbol{#1}}

\mathchardef\mhyphen="2D

\allowdisplaybreaks
\hypersetup{
    colorlinks,
    citecolor=blue4,
    linkcolor=red4,
    urlcolor=purple4,
    breaklinks=true,
}

\begin{document}
\title{Low-depth amplitude estimation via statistical eigengap estimation}

\author{Po-Wei Huang}
\email{po-wei.huang@maths.ox.ac.uk}
\affiliation{Mathematical Institute, University of Oxford, Woodstock Road, Oxford OX2 6GG, United Kingdom}
\affiliation{Quantum Motion, 9 Sterling Way, London N7 9HJ, United Kingdom}

\author{B\'alint Koczor}
\email{balint.koczor@maths.ox.ac.uk}
\affiliation{Mathematical Institute, University of Oxford, Woodstock Road, Oxford OX2 6GG, United Kingdom}
\affiliation{Quantum Motion, 9 Sterling Way, London N7 9HJ, United Kingdom}

\date{\today}
\begin{abstract}
Amplitude estimation, in its original form, is formulated as phase estimation upon the Grover iterate. Subsequent improvements to the algorithm have eliminated the need for phase estimation and introduced low-depth variants that trade speedups for lower circuit depth. We make the key observation that amplitude estimation is equivalent to estimating the energy gap of an effective Hamiltonian, whereby discrete-time evolution is generated by amplitude amplification. This enables us to develop an amplitude estimation algorithm for both Heisenberg-limited and low-depth circuit regimes, inspired by statistical phase estimation techniques developed for early fault-tolerant ground-state energy estimation. In the Heisenberg-limited regime, our approach achieves performance comparable to state-of-the-art methods while using simplified classical post-processing. In the low-depth regime, it obtains optimal query--depth tradeoffs up to polylogarithmic factors, with provable guarantees and improved empirical performance over prior approaches. The resulting protocol is ancilla-free and requires only standard Grover reflections. Due to its flexibility, generality, and robustness, we expect our approach to be a key enabler for a broad range of early fault-tolerant applications.
\end{abstract}

\maketitle
\section{Introduction}
The original amplitude estimation of \citet{brassard2002quantum} applies phase estimation~\citep{kitaev1995quantum,cleve1998quantum,nielsen2010quantum} to Grover's amplitude amplification iterate~\citep{grover1996fast}, but requires multiple controlled operations, ancilla qubits, and a quantum Fourier transform (QFT)~\citep{coppersmith1994approximate}. Literature has focused on simplifying amplitude estimation, e.g., removing the phase estimation subroutine~\citep{suzuki2020amplitude,aaronson2020quantum,grinko2021iterative,rall2023amplitude,venkateswaran2021quantum,labib2024quantum,yang2025classical} and developing ``low-depth'' algorithms that trade off speedup factors for a lower circuit depth~\citep{rall2023amplitude, giurgicatiron2022low,vu2025lowdepth}. Independently, literature concerned with ground state energy estimation in Hamiltonian simulation
focused on vastly improving conventional phase estimation by developing ``statistical'' variants of phase estimation
that only require one ancilla qubit via the Hadamard test~\citep{kitaev1995quantum,cleve1998quantum} as relevant in early fault-tolerant applications~\citep{lin2022heisenberg-limited,wan2022randomized,ding2023even,ni2023low}. Among those techniques, sampling from Gaussian filter functions has been shown to be effective for producing low-depth phase estimation algorithms~\citep{wang2023quantum,ding2023simultaneous,ding2024quantum}. We refer the reader to \cref{appRelated} for an overview of the related work.

A key observation of this work is that amplitude estimation can be recast as an energy gap estimation problem for an effective Hamiltonian generated by amplitude amplification. This perspective enables the direct application of statistical phase estimation techniques, without relying on controlled operations or ancillary qubits, to recover the underlying amplitude from expectation values of observables. Building on this connection, we develop a Gaussian-filtered amplitude estimation algorithm that combines standard amplitude amplification circuits with Loschmidt echo--style constructions~\citep{goussev2012loschmidt}, where circuit depths are sampled from a discrete Gaussian distribution and post-processed via a simple classical least-squares estimator to recover the amplitude.

Our approach achieves both Heisenberg-limited and low-depth amplitude estimation within a unified, ancilla-free framework using only standard Grover reflections and observable-based eigengap inference instead of ancilla-assisted phase estimation. This provides a simple estimation procedure with provable query--depth tradeoffs, improved robustness in low-depth regimes, and full-range applicability. These features make it particularly well suited to early fault-tolerant quantum computers~\citep{campbell2021early,katabarwa2024early}, where minimizing circuit depth, ancillary qubits, and controlled operations is essential for reducing logical noise and resource overhead.

\section{Theoretical framework}
Statistical phase estimation in the early fault-tolerant regime uses controlled time evolution operators and the Hadamard test~\citep{kitaev1995quantum,cleve1998quantum} to
measure an ancillary qubit. Our theoretical results explain why, in contrast,
state-of-the-art amplitude estimation algorithms (which do not use phase estimation) require no extra ancillary qubits, nor controlled evolutions~\citep{suzuki2020amplitude, grinko2021iterative,venkateswaran2021quantum, aaronson2020quantum,rall2023amplitude,labib2024quantum}, even in short depth regimes~\citep{giurgicatiron2022low}. Specifically, in the following, we show that the Grover iterate can be viewed as a discrete-time evolution under an effective Hamiltonian, and prove that amplitude estimation can be effectively formulated as a Hamiltonian energy gap estimation algorithm (rather than phase estimation), whereby expected values of observables are measured instead of the Hadamard test circuit.

\subsection{Quantum phase and eigengap estimation}
First, let us recollect the basics of phase estimation whereby one assumes that a $K$-dimensional Hamiltonian $H$ with eigenvectors $\{\ket{\psi_i}\}_{i\in[K]}$, and an initial state $\ket{\psi} = \sum_{i=1}^K c_i \ket{\psi_i}$ are given. By controlling the time evolution operator $U_t=e^{-iHt}$ on a single ancilla qubit, its expected value can be estimated as the time series
\begin{equation}
\braket{\psi|U_t|\psi} = \sum_{k\in [K]} p_{k} e^{-itE_k},
\end{equation}
where $p_k=|c_k|^2$ and $\sum_{k\in[K]} p_k = 1$. Classical frequency analysis of $S(t) = \braket{\psi|U_t|\psi}$ allows us to obtain eigenvalues $E_k$ that correspond to the dominant eigenvectors in the initial state.

In contrast, eigengap estimation is a resource-frugal alternative that requires no controlled time evolution or additional ancilla qubits,
whereby time-evolved expected values of an observable $O$ as
\begin{equation}
    \label{eqEigengapEst}
    \braket{\psi|U_t^\dagger OU_t|\psi} = \sum_{j, k\in [K]} q_{j,k} e^{-it(E_k-E_j)},
\end{equation}
are estimated.
Here $q_{j,k} = c_j^*c_k \braket{\psi_j|O|\psi_k}$. We select $O$ suitably such that $q_{j,k} = |q_{j,k}| e^{-i\theta_{j,k}}$ where $\theta_{j,k}$ can be inferred and $ \sum_{j,k\in[K]} |q_{j,k}| = 1$. The special case $O=\ketbra{\psi}{\psi}$ yields the usual Loschmidt echo~\citep{goussev2012loschmidt}, which can be post-processed using statistical phase estimation techniques to obtain eigengaps, given $ \sum_{j,k\in[K]} q_{j,k} = 1$.

Above, continuous time evolution $U_t=e^{-iHt}$ can be simulated with a range of techniques, including conventional~\citep{trotter1959product,suzuki1976generalized,suzuki1991general} and randomized~\citep{campbell2019random, kiumi2025te} product formulas.
However, for phase estimation, it may be more efficient to apply a discrete-time evolution $U_t = e^{-itH}$ for $t \in \mathbb{N}$ under an effective Hamiltonian---a prominent example is qubitization where $U_t = e^{-itH_{\rm eff}}$ with $H_{\rm eff}=\cos^{-1}(H/\Lambda)$~\citep{low2019hamiltonian,childs2009relationship}, where $\Lambda$ is a normalization factor. Phase estimation then reveals eigenvalues of the effective Hamiltonian $H_{\rm eff}$, from which eigenvalues of the true Hamiltonian $H$ follow via $E_{\rm eff}=\pm\cos^{-1}(E/\Lambda)$. Qubitization enables state-of-the-art ground-state energy estimation~\citep{berry2019qubitization,low2025fast} and can, in principle, also be used to perform discrete-time evolution-based statistical phase estimation~\citep{ding2024quantum}. While our approach is similar in that we apply a discrete-time evolution under an effective Hamiltonian, rather than extracting eigenvalues directly, we estimate eigengaps in the effective Hamiltonian by post-processing the expected values of observables.

\subsection{Amplitude amplification as discrete-time evolution}
Given a projector $P$, it decomposes the Hilbert space into a good subspace $\mathcal{H}_G$ with $P\ket{\phi}=\ket{\phi}$ and a bad subspace with $P\ket{\phi}=\vec{0}$. Given an arbitrary input state $\ket{\psi}$, the probability of measuring it in the good subspace is 
 $a=\sin^2(\lambda)=\|P\ket{\psi}\|^2$ and we define the post-measurement states as $\ket{\psi_g}=\frac{P\ket{\psi}}{\braket{\psi|P|\psi}}$ and $ \ket{\psi_b}=\frac{(\mathbb{I}-P)\ket{\psi}}{\braket{\psi|\mathbb{I}-P|\psi}}$.

\begin{restatable}[Effective Hamiltonian of amplitude amplification]{proposition}{ham}
    We are given the Grover iterate $Q = -(\mathbb{I}-2\ketbra{\psi}{\psi})(\mathbb{I}-2P)$ for some arbitrary input state $\ket{\psi}$ and a projector $P$ such that $a = \sin^2(\lambda) = \lVert P\ket{\psi}\rVert^2$. Then $Q$ is equivalent to a discrete-time evolution under an effective Hamiltonian $H_{\rm eff}$ where
	\begin{equation*}
			Q = e^{-i H_{\rm eff}} \quad \text{with} \quad
			H_{\rm eff} = -2 \lambda\,Y_{\psi} + \pi\,(\mathbb{I}-P-\ketbra{\psi_b}{\psi_b}).
	\end{equation*}
	In the 2-dimensional subspace
	$\mathcal{H}_\psi = \mathrm{span}\{\ket{\psi_g},\ket{\psi_b}\}$, the operator $Y_{\psi} = i\ketbra{\psi_b}{\psi_g} - i\ketbra{\psi_g}{\psi_b}$ is an effective Pauli Y matrix, and therefore $Q$ acts as a discrete Pauli Y rotation. In the subspace $\mathcal{H}_\psi$, the effective Hamiltonian has two eigenvalues which are related to the amplitude to be estimated as
	\begin{equation*}
		\pm E_{\rm eff} = \pm 2 \lambda = \pm 2  \sin^{-1}(\sqrt{a}) = \pm \cos^{-1}(1-2a),		
	\end{equation*}
	where the last term uses the same convention as qubitization.
	\label{lemHam}
\end{restatable}

We defer the proof to \cref{appLemHam} and note the following points. First, conventional amplitude estimation proceeds by applying a discrete-time evolution under the effective Hamiltonian $H_{\rm eff}$ as part of phase estimation. Second, applying $Q$ produces a nontrivial Pauli-$Y$ rotation $e^{2\lambda Y_{\psi}}$ with eigenphases $\pm\cos^{-1}(1-2a)$ in the subspace $\mathcal{H}_\psi$. Equivalently, projecting $R = U^\dagger(\mathbb{I}-2P)U$ onto the zero subspace gives $\braket{0|R|0} = 1-2a$. This provides the eigenvalues $e^{\pm i\arccos(1-2a)}$ of $U^\dagger(\mathbb{I}-2P)U(2\ketbra{0}{0}-\mathbb{I})$ per qubitization~\citep{low2017optimal}, which shares eigenvalues with $Q$ by unitary invariance. 

\begin{figure}
    \includegraphics[width=\linewidth]{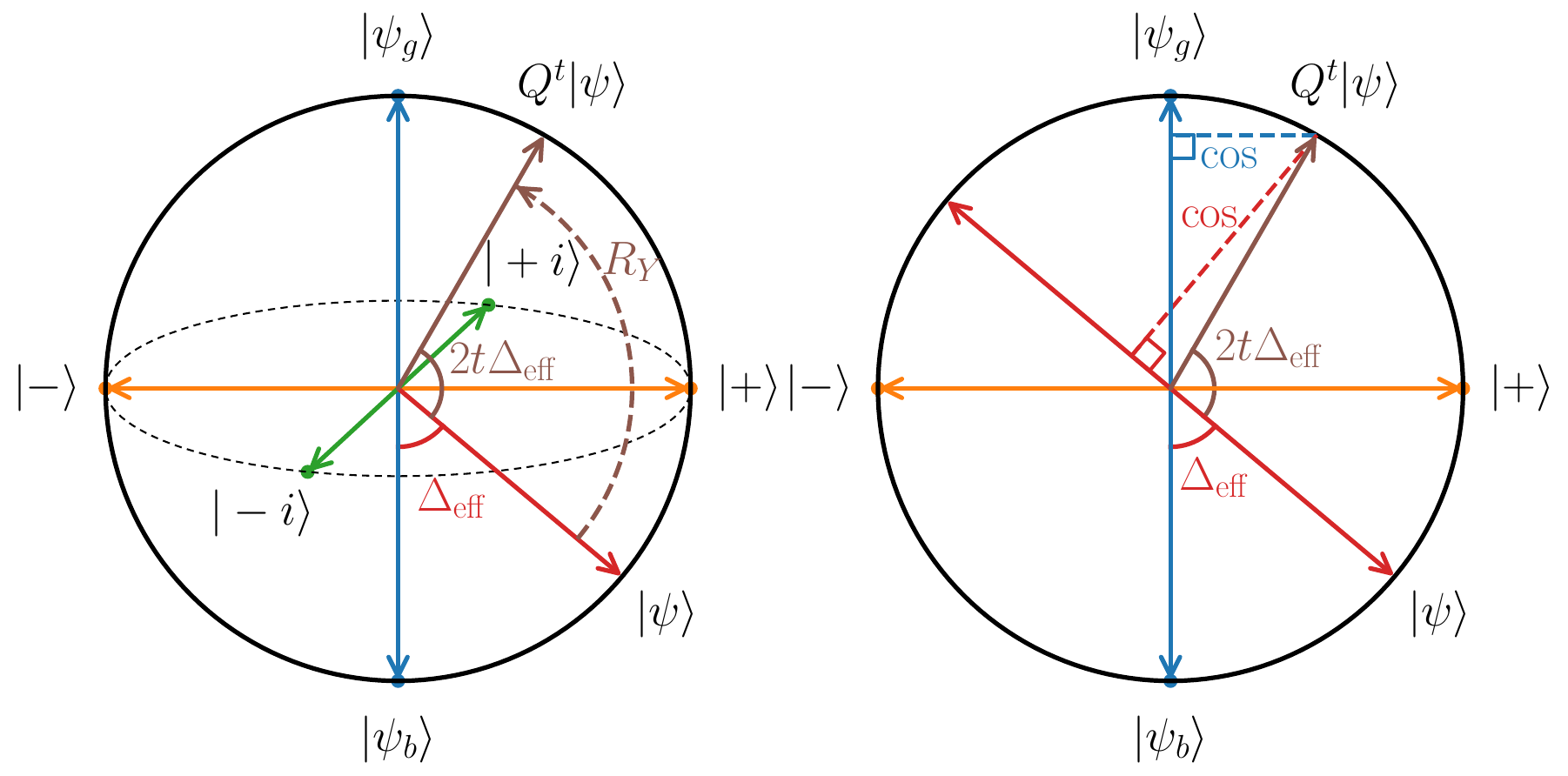}
    \caption{Bloch-sphere representation of amplitude amplification as discrete-time evolution in the two-dimensional subspace $\mathcal{H}_\psi$. The Grover iterate $Q$ acts as a rotation about the effective Pauli-$Y$ axis with angle $\Delta_{\rm eff}$ determined by the energy gap. Measurements along Pauli-like observables $O$ and $O'$ provide cosine signals generated by $\Delta_{\rm eff}$ that can be used for extraction.}
    \label{figEigengap}
\end{figure}

\begin{figure*}
    \includegraphics[width=\linewidth]{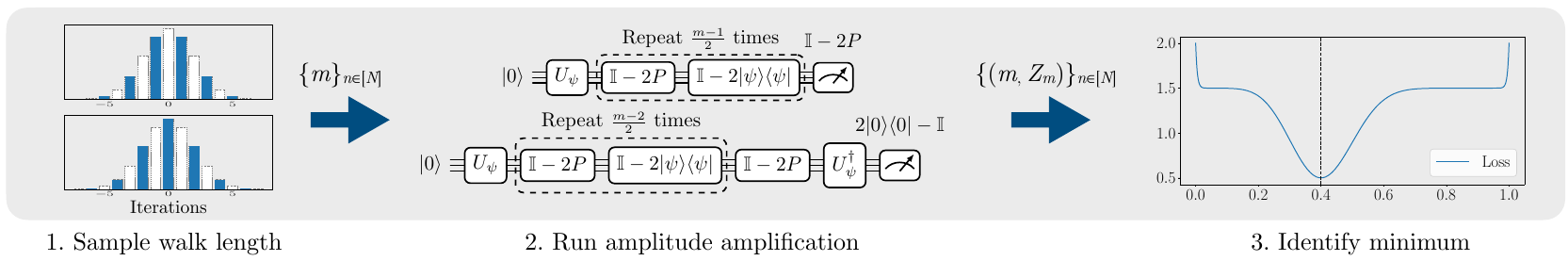}
    \caption{\emph{Algorithm.} Our algorithm follows a simple three-step protocol. First, a series of $\{m\}$ is sampled from a truncated discrete Gaussian. Second, we execute amplitude amplification circuits with $m$ applications of $U_\psi$. Lastly, we construct a target function based on the measurement result for which the peak lies on an $\epsilon$-close estimate $\widetilde a$ of the target amplitude $a= \lVert P\ket\psi\rVert^2$.}
    \label{figOver}
\end{figure*}

\subsection{Amplitude estimation as eigengap estimation}

Rather than estimating the desired amplitude via phase estimation, one can instead apply more resource-efficient eigengap estimation to extract the energy gap in the effective Hamiltonian. The main difference to conventional gap estimation is that we apply discrete-time evolution $U_t:=Q^t $ where $t\in\mathbb{N}$ and that we measure two Pauli-like observables $O = \mathbb{I} - 2P$ and the observable corresponding to a shifted Loschmidt echo $O' = 2\ketbra{\psi}{\psi} - \mathbb{I}$.
\begin{restatable}[Effective Hamiltonian eigengap estimation]{proposition}{eigap}
	Given the discrete-time evolution operator $U_t := Q^t $ under the effective Hamiltonian 
	$H_{\rm eff}$ from \cref{lemHam},
	the expected value of the Pauli-like observable $O = \mathbb{I} - 2P$
	yields the time-dependent expected value in \cref{eqEigengapEst}
	\begin{align*} \nonumber
		\braket{\psi|U_t^\dagger O U_t|\psi} &= \sum_{E_j, E_k\in \{\pm E_{\rm eff}\}} q_{j,k} e^{-i t ( E_j - E_k)}\\
		&= \frac{1}{2} e^{i(t+1/2)\Delta_{\rm eff}}  + \frac{1}{2} e^{-i(t+1/2)\Delta_{\rm eff}} 
	\end{align*}
	which depends on the effective Hamiltonian's energy gap  $\Delta_{\rm eff} =2 E_{\rm eff} = 4 \lambda = 4 \sin^{-1}(\sqrt{a})$. Alternatively, using the observable $O' = 2\ketbra{\psi}{\psi} - \mathbb{I}$, we obtain the time series
    \begin{equation*} \nonumber
		\braket{\psi|U_t^\dagger O' U_t|\psi} = \frac{1}{2} e^{it\Delta_{\rm eff}}  + \frac{1}{2} e^{-it\Delta_{\rm eff}}.
	\end{equation*}
    \label{lemEigap}
\end{restatable}
To provide an intuitive sketch for this result, we visualize the dynamics confined to the subspace $\mathcal{H}_\psi$ as one on an effective qubit Bloch sphere in \cref{figEigengap}. Under $U_t = Q^t$, the state undergoes a rotation around the Pauli-$Y$ axis in this subspace with frequency set by the relative phases of the two eigenstates, or the energy gap $\Delta_{\rm eff}$.

The role of the observables is to probe different projections of this rotation. The operator $O = \mathbb{I}-2P$ projects along the $\{|\psi_g\rangle,|\psi_b\rangle\}$ axis equivalent to a Pauli-$Z$ measurement, while $O' = 2\ketbra{\psi}{\psi}-\mathbb{I}$ projects along the initial state direction. Together, these observables would allow extraction of the energy gap $\Delta_{\rm eff}$, which encodes the amplitude $a$.

We defer the formal proof to \cref{appLemEigap}. The use of the shifted Loschmidt echo measurement of $O'$, in addition to the ``standard'' amplitude amplification circuit, is reminiscent of the use of even degree Chebyshev polynomials in amplitude estimation in the context of quantum signal processing~\citep{rall2023amplitude}. Finally, we note that the last state reflection in $Q^t$ may be omitted by replacing $Q^t$ with $(\mathbb{I}-2P)Q^{t-1}$, as the final measurement is performed on $\ket{\psi}$ and the reflection by $\ket{\psi}$ has no effect.

\section{Gaussian-filtered ancilla-free amplitude estimation}

In this section, we present the Gaussian Least Squares Amplitude Estimation (GLSAE) algorithm and analyze its performance. We first describe the sampling strategy, circuit implementation, and post-processing procedure that produces the amplitude estimate. We then establish its theoretical guarantees, including its Heisenberg-limited scaling, optimal query--depth tradeoffs at low depths, as well as additional mechanisms to ensure reliable estimation across the full range of amplitudes.

\subsection{Algorithm and estimation procedure}
Our algorithm follows a simple three-step protocol as illustrated in \cref{figOver}:
\begin{enumerate}
    \item Sample iteration counts $m$ from a sampling filter function $\widetilde p_T$.
    \item Apply amplitude amplification circuits described in \cref{lemEigap} corresponding to the sampled iteration $m$ and obtain measurement outcome $Z_m$.
    \item Minimize a loss function $\bar{\mathcal{L}}$ to find the underlying eigengap $\Delta_{\rm eff} = 4\lambda$ and use $\lambda$ to obtain the outcome $a = \sin^2(\lambda)$.
\end{enumerate}

The expected value of the time-evolved observable measurement in \cref{lemEigap} is trivially related to the discrete-time
signal $S(m)=\cos(2\lambda m)$ for $m \in \mathbb{Z}$, i.e., averaging individual shots $Z_m\in\{\pm1\}$ of observable measurements yields the mean $\mathbb{E}[Z_m|m]=\cos(2\lambda m)$. Specifically, for odd $m$, we measure $O = \mathbb{I}{-}2P$ per \cref{lemEigap} and set $m = 2(t{+}\frac{1}{2})$, while for even $m$, we measure $O = 2\ketbra{\psi}{\psi}{ -} \mathbb{I}$ and set $m=2t$. The circuit depth is directly proportional to $m$, which we can limit by sampling the evolution times $m \le M = \mathcal{O}(T)$ from a truncated discrete Gaussian distribution $\widetilde{p}_T$ which has standard deviation $T$---for practical purposes where precision is set to around $10^{-5}-10^{-6}$, $M=5T$ is a reasonable approximation. In particular, we use the truncated discrete Gaussian as follows:
\begin{equation}
    \widetilde{p}_T(m) = \begin{cases}
        \frac{1}{\sqrt{2\pi}T\widetilde{\mathcal{Z}}}\exp\left(-\frac{m^2}{2T^2}\right) & 1 \le |m| \le 5 T\\
        0 & |m| > 5 T\\
        1 - \sum_{m'\ne 0} \widetilde p_T(m') & m=0
    \end{cases}
\end{equation}
where $\widetilde{\mathcal{Z}} = \sum_{m=-15 T}^{15 T} \frac{1}{\sqrt{2\pi}T}\exp\left(-\frac{m^2}{2T^2}\right)$. 

For the sampled evolution time $m$, we measure the corresponding amplitude amplification circuit of depth $\mathcal{O}(m)$ and record the measurement outcome $Z_m$---the measured $Z_m$-s would then act as a sparse, noisy cosine signal that induces a form of sparse discrete-time cosine transform~\citep{oppenheim2009discrete} over the discrete Gaussian to form a periodic Gaussian. 

Inspired by statistical phase estimation techniques~\citep{wang2023quantum,ding2023even,ding2023simultaneous}, we consider the loss function as the mean square error between the observed signal and the ideal cosine signal as
\begin{equation}
    \widetilde{\mathcal{L}}(\theta) = \frac{1}{N} \sum_{m\sim \widetilde{p}_T} \big(Z_m - \cos(2\theta m)\big)^2,
\end{equation}
whose minimum recovers the amplitude as $\lambda \approx \arg\min_{\theta \in [0, \pi/2]}\widetilde{\mathcal{L}}(\theta)$. We term this approach Gaussian Least Squares Amplitude Estimation (GLSAE), which is able to yield Heisenberg-limited scaling of total quantum resources (space-time volume). Moreover, by tuning the discrete Gaussian's standard deviation $T$ and truncation threshold $M \in \mathcal{O}(T)$ (equivalently, the maximum circuit depth), we obtain a trade-off between circuit depth and the number of samples $N$---this enables low-depth amplitude estimation. The full algorithm can be found as described in \cref{algoGLSAE}. We analyze the estimator’s performance and its behavior near the boundary in the following subsection.

\begin{algorithm}[t]
\caption{Gaussian Least Square Amplitude Estimation (GLSAE)}
\label{algoGLSAE}
\Indm
\KwIn{State $\ket{\psi}$, Projector $P$, Grover iterate $Q$, Gaussian width $T$, Maximum circuit depth $M$, Number of parallel queries $N$, Precision $\epsilon$}
\KwOut{Estimate of $a = \lVert P\ket{\psi}\rVert^2$}
\Indp
\DontPrintSemicolon
\vspace{1ex}
\nlnonumber \textbf{\hspace{-1.4em}I. Sample}\;
Sample $N$ iterations $m$ from truncated discrete Gaussian 
$$\widetilde{p}_T(m) = \begin{cases}
        \frac{1}{\sqrt{2\pi}T\widetilde{\mathcal{Z}}}\exp\left(-\frac{m^2}{2T^2}\right) & 1 \le |m| \le M\\
        0 & |m| > M\\
        1 - \sum_{m'\ne 0} \widetilde p_T(m') & m=0
    \end{cases}$$\;
\vspace{-3ex}
\nlnonumber \textbf{\hspace{-1.4em}II. Circuit execution}\;
\For {each sampled $m$ out of $N$ samples}{
    \uIf{$|m| = 2t + 1$, $t\in\mathbb{Z}_{\ge 0}$}{
        Measure observable $\mathbb{I}-2P$ on state $Q^t\ket{\psi}$ and obtain $Z_m \in \{\pm 1\}$
    }
    \ElseIf{$|m| = 2t$, $t\in\mathbb{N}$}{
        Measure observable $2\ketbra{\psi}{\psi}-\mathbb{I}$ on state $(\mathbb{I}-2P)Q^{t-1}\ket{\psi}$ and obtain $Z_m \in \{\pm 1\}$
    }
}
\vspace{1ex}
\nlnonumber \textbf{\hspace{-1.4em}III. Optimization}\;
Obtain estimate $\widetilde \lambda$ of $\lambda$ from
$$\begin{aligned}
    \widetilde \lambda &= \argmin_{\theta \in [0, \frac{\pi}{2}]}\widetilde{\mathcal{L}}(\theta) \\
    &= \argmin_{\theta \in [0, \frac{\pi}{2}]}\frac{1}{N} \sum_{m \sim \widetilde p_T} \left(Z_m - \cos(2\theta m)\right)^2
\end{aligned}$$
via a two-level grid search for $\chi^*$ across 
$$\theta \in \left\{\frac{\pi\chi}{2M}: \chi \in [M]\right\}$$
 followed by a more fine-grained search for $\eta^*$ across
 \[\theta \in \left\{ \frac{\pi\chi^*}{2M} + \frac{\eta\epsilon}{2} : \eta \in \left[-\frac{8\pi}{M\epsilon}..\frac{8\pi}{M\epsilon}\right]\right\}\]\;
\vspace{-3ex}
\uIf {$\widetilde \lambda \in [\frac{1}{T},\frac{\pi}{2}-\frac{1}{T}]$}{
\Return $\widetilde a = \sin^2(\widetilde \lambda)$
}
\Else{
    Perform $N$ additional measurements of $\|P\ket{\psi}\|^2$ and obtain 
    $Z' \sim {\rm Bernoulli}(a)$\;
    \Return $\widetilde a = \frac{1}{N} \sum_N Z'$\;
}
\end{algorithm}

\subsection{Performance analysis and theoretical guarantees}
We now analyze the performance of GLSAE and its dependence on the sampling and circuit parameters. When the sampling Gaussian width $T$ is sufficiently large, the induced loss landscape is sharply concentrated around the true parameter, enabling accurate recovery of $\lambda$. As the circuit depth decreases, however, this landscape becomes less well-resolved, particularly for amplitudes $a$ close to $0$ or $1$.

As the cosine signal produced in \cref{lemEigap} has two frequency peaks corresponding to minima of $\widetilde{\mathcal{L}}$ at $\theta=\pm\lambda$, a low standard deviation $T$ of our Gaussian may cause a significant broadening of these two peaks. When the broadening is comparable to the distance between the two peaks, the minimum of $\widetilde{\mathcal{L}}$ is difficult to discern in the presence of sampling noise due to finite repetition. Consequently, as formalized below, the estimable range of $a$ is constrained by $T$ (and consequently $M$) via the tuning $\beta$ and range $\zeta$ parameters. 

\begin{restatable}[Gaussian Least Squares Amplitude Estimation]{theorem}{main}
\label{thmMain}
Fix a tuning parameter $\beta \in [0,1]$ and a target precision $\epsilon > 0$. For some range parameter $\zeta \in \widetilde{\mathcal{O}}(\epsilon^{2-2\beta})$ dependent on circuit depth, there exists a quantum algorithm \emph{GLSAE} that outputs, with high probability, an estimate $\widetilde a$ of $a = \lVert P\ket{\psi}\rVert^2$, for any $a \in [\zeta, 1-\zeta]$, such that 
\begin{equation*}
\lvert \widetilde a - a \rvert \le 2\sqrt{a(1-a)}\epsilon + \epsilon^2.
\end{equation*}
Moreover, \emph{GLSAE} requires a maximum circuit depth $M \in \widetilde{\mathcal{O}}(\epsilon^{-1+\beta})$ and $N \in \widetilde{\mathcal{O}}(\epsilon^{-2\beta})$ parallel circuit samples.
\end{restatable}

The proof for this theorem, detailed in \cref{appProofBookkeep,appProofMain}, utilizes strong convexity and smoothness properties of the periodic Gaussian produced from the discrete-time cosine transform---derived in \cref{appGauss}---to obtain quadratic upper and lower bounds on the loss landscape around its minimum. These properties, when coupled with Hoeffding bounds on the number of samples, yield a quadratic inequality governing the estimation error, which, when solved, reproduces the Zalka--Burchard lower bound of $M^2N \in \Omega(\epsilon^{-2})$~\citep{zalka1999grover,burchard2019lower} up to polylogarithmic factors found in low-depth amplitude estimation protocols~\citep{giurgicatiron2022low}, as well as prior bounds found in quantum metrology and phase estimation~\citep{giovannetti2006quantum,wang2019accelerated,tong2021tight}. Remarkably, our proof for discrete-time effective Hamiltonian evolution yields a single bound that holds uniformly across all circuit depths---in contrast to prior work, which analyzed separate regimes depending on circuit depth and spectral gap assumptions, even in the context of conventional Hamiltonian simulation~\citep{ding2023simultaneous,ding2024quantum}.

\cref{thmMain} depends on a tunable free parameter $\beta$, which affects the polynomial scaling of the runtime with respect to a target precision $\epsilon$ such that the number of queries to the state preparation unitary can be interpolated between the Heisenberg and standard quantum limit, i.e., precision scaling between $\widetilde{\mathcal{O}}(\epsilon^{-1})$ and $\widetilde{\mathcal{O}}(\epsilon^{-2})$. 
In contrast, the following corollary focuses on the special case when our approach applies to all $a$ and achieves the Heisenberg limit of both $\mathcal{O}(\epsilon^{-1})$ circuit depth and total $\widetilde{\mathcal{O}}(\epsilon^{-1})$ state preparation queries.
\begin{restatable}[Heisenberg-limited GLSAE]{corollary}{corHeis}
Fixing $\beta =0$, with high probability, \emph{GLSAE} outputs an estimate $\widetilde a$ of $a = \lVert P\ket{\psi}\rVert^2$, for any $a \in [0,1]$, such that 
\begin{equation*}
\lvert \widetilde a - a \rvert \le 2\sqrt{a(1-a)}\epsilon + \epsilon^2.
\end{equation*}
\emph{GLSAE} requires a maximum circuit depth $M \in \mathcal{O}(\epsilon^{-1})$ and $N \in \mathcal{O}(\log(\epsilon^{-1}))$ parallel circuit samples.
\label{corHeis}
\end{restatable}
We further note that in moderate-depth regimes where $\beta \in (0, \frac{1}{2}]$, we have $\zeta \in \mathcal{O}(\epsilon^{2-2\beta}) \subseteq \mathcal{O}(\epsilon)$, and thus GLSAE is capable of retrieving $a$ up to a slightly weaker guarantee of achieving an additive precision of $\mathcal{O}(\epsilon)$. 

In the very low-depth regime where $\beta \ge \frac{1}{2}$, while GLSAE can still be used to provide an estimate up to the amplitude-dependent precision shown in \cref{thmMain} when $a\in [\zeta, 1-\zeta]$, the range is much more limited. In this regime, the lowest depth that we can achieve for true amplitude $a$ is then $\widetilde{\mathcal{O}}(a^{-\frac{1}{2}})$, providing an improvement over~\citet{rall2023amplitude}'s requirement of $\widetilde{\mathcal{O}}(a^{-1})$. 

Despite this restricted range in \cref{thmMain}, we can still reliably identify when GLSAE operates outside its valid regime. In particular, although \cref{thmMain} assumes $a \in (\zeta, 1-\zeta)$, the output itself provides a certificate for correctness. We formalize this through the following corollary, which allows us to reject unreliable estimates based on the observed outcome.

\begin{restatable}[Acceptance certification for GLSAE]{lemma}{certi}
\label{corCerti}
Fix a tuning parameter $\beta \in [0,1]$, a target precision $\epsilon > 0$. Let $\widetilde a$ be the estimate returned by GLSAE with depth and sample counts set by $\beta$ and $\epsilon$, and define the following acceptance rule for an existing range parameter $\zeta \in \widetilde{\mathcal{O}}(\epsilon^{2-2\beta})$ dependent on circuit depth:
\begin{equation*}
\text{accept if } \widetilde a \in [16\zeta,\,1-16\zeta],\; \text{otherwise reject.}
\end{equation*}
With high probability, we have
$|\widetilde a - a| \le 2\sqrt{a(1-a)}\epsilon +\epsilon^2$ for accepted estimates.
\end{restatable}

We detail the proof in \cref{appProofCert}. This lemma corresponds to the last \emph{if} statement that calls for additional samples when the main output fails. Thus, even in the low-depth regime where the estimable range is limited, GLSAE remains self-certifying. Unreliable estimates are rejected, while accepted estimates retain the guarantees of \cref{thmMain}. 

Lastly, to resolve the ambiguity problem caused by the access to only cosine signals, and achieve estimation over the full range of amplitudes, we provide a slightly weaker guarantee that provides an estimate with uniform additive error $\epsilon$, rather than the amplitude-dependent error scaling typical of amplitude estimation as shown in \cref{thmMain} and \cref{corHeis}.  

\begin{restatable}[Full-range GLSAE]{theorem}{plus}
Fix a tuning parameter $\beta \in [0,1]$, a target precision $\epsilon > 0$. Then, with high probability, there exists a quantum algorithm \emph{GLSAE} that outputs an estimate $\widetilde a$ of $a = \lVert P\ket{\psi}\rVert^2$, for any $a \in [0,1]$, such that 
\begin{equation*}
\lvert \widetilde a - a \rvert \le \epsilon.
\end{equation*}
Moreover, \emph{GLSAE} requires a maximum circuit depth $M \in \widetilde{\mathcal{O}}(\epsilon^{-1+\beta})$ and $N \in \widetilde{\mathcal{O}}(\epsilon^{-2\beta})$ parallel circuit samples.
\label{thmPlus}
\end{restatable}

We detail the proof in \cref{appProofPlus}. To obtain the estimate when the main GLSAE subroutine fails, one simply performs an additional $N$ projective measurements of $\ket{\psi}$ onto $P$ for estimates rejected by \cref{corCerti}. This is essentially sampling from a Bernoulli distribution of expectation value $a$. At near-boundary terms within $\mathcal{O}(\epsilon^{2-2\beta})$ of $0$ and $1$, the variance of the Bernoulli distribution approaches zero, and $\mathcal{O}(\epsilon^{-2\beta})$ samples  suffice to produce an $\epsilon$-close estimate of $a$. This process covers when the main GLSAE routine fails, and overall achieves low-depth amplitude estimation over the full amplitude range.

\begin{figure*}
	\includegraphics[width=\linewidth]{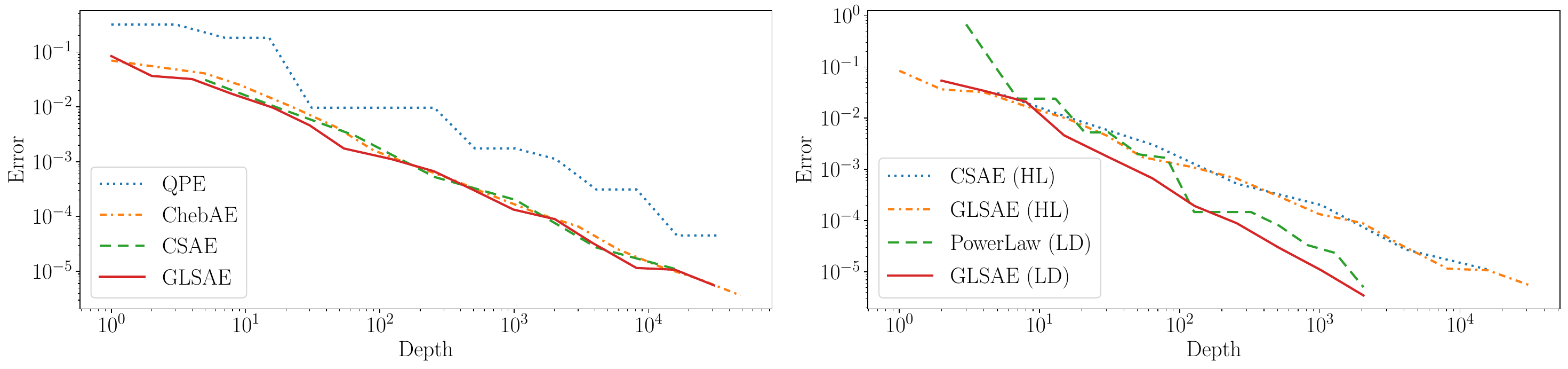}
	\caption{
		\emph{Comparison of the circuit depth of near-Heisenberg-limited amplitude estimation algorithms (left) and low-depth amplitude estimation (right).}
		GLSAE achieves state-of-the-art performance, comparable to ChebAE and CSAE. In the low-depth regime (right), we implement the case
		with circuit depth and sample complexity equivalent to uniform time sampling where $M, N \in \widetilde{\mathcal{O}}(\epsilon^{-2/3})$. Notably, our algorithm improves over prior art like Power Law AE in such low-depth regimes by half an order of magnitude, and much more prominently in near-constant-depth regimes.}
	\label{figComp}
\end{figure*}

\section{Numerical results and benchmarking}
We benchmark our algorithm against a broad range of prior amplitude estimation techniques, including ``textbook'' amplitude estimation that uses phase estimation~\citep{brassard2002quantum}, ChebAE from quantum signal processing~\citep{rall2023amplitude}, and CSAE as a modified version of ESPRIT~\citep{labib2024quantum}. We additionally benchmark against adaptations of other ideas in signal processing and statistical phase estimation, which we detail in \cref{appNum}.

As shown in \cref{figComp} (left), GLSAE achieves comparable performance to the prior state-of-the-art, ChebAE and CSAE, at the Heisenberg limit. Beyond empirical performances, we also highlight some comparisons of the actual execution of these algorithms. First, like CSAE, our algorithm admits high parallelization, in contrast to ChebAE, due to its non-iterative structure. Second, it features a substantially simpler construction compared to CSAE, which relies on costly classical post-processing to achieve Heisenberg-limited scaling by constructing a virtual uniform signal array via iterative outer products of measurement results sampled from an exponentially increasing sequence.

In the low-depth regime, \cref{figComp} (right) shows that GLSAE improves over prior methods such as Power Law AE~\citep{giurgicatiron2022low}, which has been shown to outperform QoPrime AE in comparable regimes~\citep{giurgicatiron2022low,giurgicatiron2022lowexp}. In particular, GLSAE approximately achieves a half-order-of-magnitude improvement. Additionally, unlike Power Law AE, GLSAE applies in both low-depth and Heisenberg-limited regimes. We also note that prior work by \citet{vu2025lowdepth} achieves comparable theoretical guarantees, but requires additional auxiliary qubits, whereas our method operates using only standard amplitude amplification circuits.

Finally, we examine the query--depth tradeoff of GLSAE at a fixed target precision in \cref{figInvar}. We compare different $(\mathcal{D},\mathcal{Q})$ pairs, where $\mathcal{D}$ denotes the maximum depth and $\mathcal{Q}$ the total number of queries, rather than the parallel circuit count $N$ and depth $M$ used in the main proof, to account for practically relevant quantum resources and align with prior methods. Indeed, as long as the product $\mathcal{D}\mathcal{Q}$ is fixed, the estimation error is approximately constant (identical height bars in \cref{figInvar}). In contrast to prior low-depth amplitude estimation methods, which require restrictive depth schedules~\citep{giurgicatiron2022low}, strong regularity assumptions~\citep{giurgicatiron2022low}, or constraints on estimator bias~\citep{rall2023amplitude,vu2025lowdepth}, GLSAE achieves the query--depth tradeoff $\mathcal{D}\mathcal{Q} \in \widetilde{\mathcal{O}}(\epsilon^{-2})$ with a efficient single-variable post-processing optimization protocol.

\begin{figure}
\includegraphics[width=\linewidth]{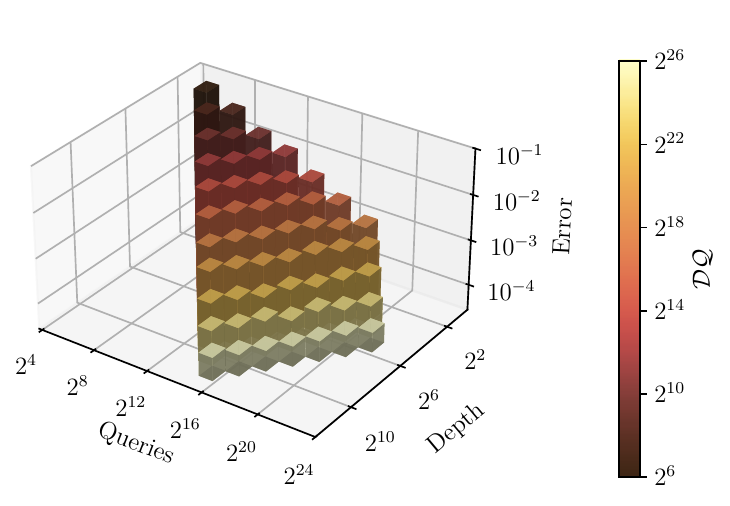}
\caption{\emph{Query--depth invariance of low-depth amplitude estimation.} We show query--depth invariance of running GLSAE on low-depth amplitude amplification circuits. Varying the depth $\mathcal{D}$ and query $\mathcal{Q}$ counts yields estimates of approximately identical precision as long as the product $\mathcal{D}\mathcal{Q}$ is fixed.}
\label{figInvar}
\end{figure}

\section{Discussion}
In this work, we recast amplitude estimation as an eigengap estimation problem and develop a Gaussian-filtered algorithm that achieves both Heisenberg-limited and low-depth resource scalings. Beyond their conceptual simplicity, these methods offer practical advantages over existing techniques, with our numerical results demonstrating comparable performance to prior algorithms that require more complicated classical post-processing. Furthermore, our low-depth algorithm improves performance when compared to prior art and supports a query--depth tradeoff in terms of runtime speedups.

From a theoretical perspective, our results clarify why modern amplitude estimation algorithms~\citep{suzuki2020amplitude,grinko2021iterative,venkateswaran2021quantum, aaronson2020quantum,rall2023amplitude,labib2024quantum} can avoid controlled Grover iterates and ancillary qubits, as well as using phase estimation as a subroutine. Instead of implementing Hadamard tests or textbook phase estimation~\citep{cleve1998quantum,nielsen2010quantum} to extract the phase, one can probe cosine signals of the eigengap generated by amplitude amplification circuits. Rather than estimating absolute eigenphases via an auxiliary qubit, the eigengap formulation extracts relative spectral information directly from expectation values of observables, allowing the gap, and, by extension, the amplitude, to be recovered via classical post-processing.

The advantage of removing the ancillary qubit required of the Hadamard test is the elimination of additional controlled operations to implement reflections about the projector $P$ and the initial state $\ketbra{\psi}{\psi}$. This is especially relevant when $P$ corresponds to a more complex object, such as the membership oracle for quantum counting~\citep{brassard1998quantum,aaronson2020quantum} or block encodings of molecular observables~\citep{steudtner2023fault}.

On the other hand, this benefit is less pronounced for the reflection about $\ketbra{\psi}{\psi}$ as the realization of the reflection operator $U(2\ketbra{0}{0} - \mathbb{I})U^\dagger$ typically already requires a multi-qubit Toffoli spanning the full register. However, while amplitude estimation requires this global reflection, this does not impede its suitability in early fault-tolerant regimes. Recent work~\citep{gosset2025multi} has shown that such multi-qubit Toffoli gates can be implemented with exponentially reduced $T$-count when probabilistic mixtures are allowed. This is compatible with our algorithm, as it already incorporates sampling over circuit depth.

A technical subtlety arises in our methods due to its reliance on cosine signals, which introduces an ambiguity in the underlying frequency when $\lambda$ is near $0$ or $\frac{\pi}{2}$. This affects only boundary regimes, where $a \le \zeta$ or $a \ge 1-\zeta$, but in this regime, the amplitude is effectively known \emph{a priori}. We can then use this knowledge to resolve the ambiguity with a small number of additional measurements. A similar issue appears in prior work by \citet{rall2023amplitude}, although the affected boundary regime is wider in their setting, making comparable corrections more challenging.

Other low-depth amplitude estimation methods~\citep{giurgicatiron2022low} do not explicitly exhibit this behavior, although they rely on additional assumptions or constructions. In particular, the approach of \citet{vu2025lowdepth} achieves comparable theoretical guarantees using aggregated unbiased amplitude estimation~\citep{cornelissen2023sublinear}, where the amplitude is encoded into a phase and subsequently re-encoded back into an amplitude. While effective, this procedure requires additional auxiliary registers for the amplitude-to-phase and phase-to-amplitude transformations~\citep{gilyen2019optimizing}. In contrast, our method operates directly within the amplitude amplification framework using only standard Grover reflections, avoiding such auxiliary constructions with increased qubit consumption.

Overall, our results suggest that the GLSAE algorithm, based on an eigengap formulation of amplitude estimation, offers a flexible and resource-efficient alternative to phase-estimation--based approaches, particularly in early fault-tolerant regimes where circuit depth and additional controlled operators are key constraints. We therefore expect our framework to be a useful and practical primitive for applications such as expected value estimation, mean estimation, observable estimation, and further related tasks \citep{montanaro2015quantum,kothari2023mean,aaronson2020quantum,cleve1998quantum,steudtner2023fault}.

\begin{acknowledgments}
The authors thank Hamza Jnane, Joshua Dai, Tom Bromley, and Simon Benjamin for valuable discussions. They also further thank Jona Erle for carefully reviewing the theorems and proofs in the original preprint and for identifying inconsistencies that have been corrected in the present manuscript.

PWH acknowledges support from the Engineering and Physical Sciences Research Council (EPSRC) Doctoral Training Partnership (EP/W524311/1) with a CASE Conversion Studentship in collaboration with Quantum Motion. PWH further acknowledges support from the Ministry of Education, Taiwan, for a Government Scholarship to Study Abroad (GSSA) and St. Catherine's College, University of Oxford, for an Alan Tayler Scholarship. BK thanks UKRI for the Future Leaders Fellowship Theory to Enable Practical Quantum Advantage (MR/Y015843/1). BK also acknowledges funding from the EPSRC project Robust and Reliable Quantum Computing (RoaRQ, EP/W032635/1). For the purpose of Open Access, the author has applied a CC BY public copyright license to any Author Accepted Manuscript version arising from this submission.
\end{acknowledgments}

\paragraph*{Code and Data Availability.}The source code and data for numerics can be found at \url{https://github.com/georgepwhuang/qae_with_spe}.
    
\bibliography{main}
    
\clearpage
\appendix

\crefalias{section}{appendix}
\crefalias{subsection}{appendix}

\counterwithin{theorem}{section}
\counterwithin{proposition}{section}
\counterwithin{lemma}{section}
\counterwithin{assumption}{section}
\counterwithin{equation}{section}
\counterwithin{definition}{section}

\onecolumngrid
\begin{center}
\noindent{\large\bfseries Appendices for ``Low-depth amplitude estimation via statistical eigengap estimation''}
\end{center}
\appendixtableofcontents
\vspace{3em}
\twocolumngrid

\section{Related work}
\label{appRelated}
We now provide a summary of advancements in amplitude estimation and phase estimation in recent years.

\subsection{Amplitude estimation}
Amplitude estimation~\citep{brassard2002quantum} is a staple in designing quantum algorithms with Grover-type speedups~\citep{grover1996fast}, due to its formulation as a quadratic speedup over classical Monte Carlo algorithms~\cite{montanaro2015quantum}. Since its introduction, amplitude estimation has been used to provide quantum algorithms for various applications such as finance~\citep{rebentrost2018quantum,stamatopoulos2020option,egger2021credit,woerner2019quantum,chakrabarti2021threshold}, quantum chemistry~\citep{kassal2008polynomial, knill2007optimal,obrien2022efficient,huang2025fullqubit}, game theory~\citep{vanapeldoorn2019quantum,huang2024quantum}, machine learning~\citep{wiebe2015quantum,wiebe2016quantum,kerenidis2019qmeans,kerenidis2020qcnn}, and optimization~\citep{kerenidis2020qgd,kerenidis2020qipm,apeldoorn2020quantum,sidford2023quantum}.

Following the original formulation, recent work has demonstrated that quantum amplitude estimation can achieve Heisenberg-limited performance without relying on quantum phase estimation. \citet{suzuki2020amplitude} first showed this by combining amplitude amplification with maximum likelihood estimation, removing the need for both controlled amplification operators and ancilla qubits. This work has been followed by multiple others that discard the use of phase estimation, including adaptive approaches~\citep{aaronson2020quantum} and iterative approaches~\citep{grinko2021iterative,rall2023amplitude}, as well as parallelizable approaches that apply classical post-processing~\citep{venkateswaran2021quantum,labib2024quantum,yang2025classical}. 

Current state-of-the-art among these methods includes ChebAE from quantum signal processing~\citep{rall2023amplitude} and CSAE from a modified version of the classical signal processing algorithm ESPRIT (Estimation of Signal Parameters via Rotational Invariance Techniques)~\citep{labib2024quantum}, which we include in the numerical comparisons to our algorithm. Chebyshev amplitude estimation, or ChebAE, modifies the Iterative Quantum Amplitude Estimation (IQAE)~\citep{grinko2021iterative} such that the algorithm iteratively finds the higher Chebyshev polynomial of $a$ implementable by Grover iterates or quantum signal processing (QSP) circuits~\citep{low2017optimal} for measurement to shrink the confidence interval of the amplitude estimate $\widetilde a$. Amplitude estimation from classical signal processing, or CSAE, on the other hand, opts for an approach where measurements on an exponentially incremental sequence are used to interpolate and generate virtual signals that form a uniform signal grid via iterative outer products of the measurement and virtual signals from the previous iteration~\citep{pillai1985new}. After generating the uniform linear array of signals, a modified version of the classical signal processing algorithm ESPRIT for direction of arrival problems~\citep{roy1989esprit,liu2015remarks} is applied to recover the associated phase $\lambda$ from the input signals.

Other efforts have been made in what is known as ``low-depth'' amplitude estimation algorithms. \citet{giurgicatiron2022low} provided two such algorithms that showed that it is possible to trade off some of the quadratic speedup obtained by amplitude estimation for a lower circuit depth, paving the way for implementation in the early fault-tolerant regime~\citep{campbell2021early, katabarwa2024early,giurgicatiron2022lowexp}. Their first algorithm, Power Law AE, uses power law schedules combined with \citet{suzuki2020amplitude}'s maximum likelihood estimation amplitude estimation to obtain low-depth circuits, providing an algorithm with practical functionality, but it has theoretical guarantees that are conditioned on strong regularity conditions that may not necessarily be present. The second algorithm, QoPrime AE amplitude estimation pieces estimates from several low-depth circuits together with the Chinese Remainder Theorem to obtain the value of $a$, but is restricted to discrete circuit depths of fractional powers of the inverse error. Though the QoPrime AE algorithm provides concrete theoretical guarantees, its performance is inferior to Power Law AE~\citep{giurgicatiron2022low, giurgicatiron2022lowexp}. Further, neither algorithm can be extended to work for the Heisenberg-limited regime. 

\citet{rall2023amplitude} later proposed a low-depth algorithm that iteratively constructs (low-depth) semi-Pellian polynomial transformations of the encoded amplitude $a$ by QSP and, from the measurement results, iteratively shrinks the confidence interval of the amplitude estimate $\widetilde a$. Their algorithm has the caveat where it only operates when the circuit depth is of $\Omega(a^{-1})$ where $a$ is the true amplitude. \citet{vu2025lowdepth} applied Monte Carlo estimation on top of separate individual runs of amplitude estimation with low circuit depth to aggregate the estimates, but requires the individual amplitude estimation algorithms to be unbiased~\citep{apeldoorn2023quantum,cornelissen2023sublinear,rall2023amplitude}. \citet{oshio2025near} showed that it is possible to retain near Heisenberg-limited performance at low depths by preparing global GHZ states and parallelizing QSP circuits~\citep{martyn2025parallel}, but at a cost of high simultaneous qubit consumption and high initial all-to-all entanglement.

Lastly, Gaussian filtering was previously discussed in the context of amplitude estimation by \citet{yang2025classical}, but the filter function is computed purely on classical computers and does not apply importance sampling of quantum circuit executions, and thus does not reach Heisenberg-limited scaling.

\subsection{Phase estimation}
While efforts have been made in amplitude estimation to remove phase estimation as a subroutine, phase estimation itself has been vastly simplified to versions that only require one ancilla qubit. These versions utilize the Hadamard test~\citep{kitaev1995quantum,cleve1998quantum} instead of the full circuit including QFT as an early fault-tolerant alternative~\citep{nelson2024assessment,kiss2025early}. While \citet{kitaev1995quantum}'s original algorithm also featured a single ancilla qubit, the algorithm required exact eigenstates as input and lacks the robustness feature in more recent phase estimation algorithms. \citet{lin2022heisenberg-limited} proposed a phase estimation algorithm that achieved the Heisenberg limit by constructing the cumulative distribution function of the spectral measure, and identifies ``jumps'' in the continuity as eigenvalues of the evolved Hamiltonian. Further work has been extended on what is now known as ``statistical phase estimation'' by introducing an improved Fourier series approximation and randomized Hamiltonian simulation~\citep{wan2022randomized}. By introducing Gaussian filter functions to sample evolution times, \citet{wang2023quantum} were able to achieve low-depth circuits for phase estimation that traded the speedup of phase estimation with circuit depth. Alternative approaches~\citep{ding2023even,ni2023low} inspired by robust phase estimation~\citep{higgins2009demonstrating,kimmel2015robust,belliardo2020achieving,russo2021evaluating,matsuzaki2021direct} that applied multi-level search for eigenvalues were able to produce similar low-depth results. Further development included algorithms that achieved further speedups on the dependency of overlap of the initial state~\citep{dong2022ground-state}, as well as its low-depth counterpart~\citep{wang2025efficient}.

Similar approaches have also been applied to a generalized version of the phase estimation that involved extracting multiple eigenvalues from a single input state and Hamiltonian evolution operator. \citet{somma2019quantum} was the first to approach the task with the Hadamard test and treated the problem as a time-series analysis. Later works approached the generalized task by combining ideas from classical signal processing, including the matrix pencil method~\citep{dutkiewicz2022heisenberg,hua1990matrix,sarkar1995using}, ESPRIT~\citep{stroeks2022spectral, li2023adaptive, roy1989esprit,ding2024esprit}, and compressed sensing~\citep{yi2024quantum, castaldo2025heisenberg,candes2006near-optimal,candes2006robust}. Building upon Gaussian filtering and combining with least squares minimization, \citet{ding2023simultaneous} were able to provide low-depth algorithms for the multiple eigenvalue estimation problem, which was later improved by \citet{ding2024quantum} with a search method similar to orthogonal pursuit matching~\citep{pati1993orthogonal}.

Lastly, eigengap estimation and ancilla-free phase estimation algorithms that discard the use of controlled unitaries and the ancilla qubit to have been explored by preparing equal superpositions as the initial state~\citep{russo2021evaluating,matsuzaki2021direct}, injecting Haar random unitaries and obtaining the survival probability~\citep{zintchenko2016randomized}, combining real-time and imaginary-time evolution~\citep{yang2024phase}, utilizing classical shadows~\citep{chan2025algorithmic}, introducing Gaussian filters~\citep{yang2024resource}, and adapting phase retrieval algorithms~\citep{clinton2024quantum}, with various caveats, limitations, and runtime sacrifices attached to each algorithm.

\section{Eigengap-based formulation of amplitude estimation}
In this section, we expand on the details of the theoretical background of the main text and discuss relevant proofs. 

\subsection{Proof of amplitude amplification as discrete-time evolution}
\label{appLemHam}
As mentioned in the main text, amplitude amplification can be written as a discrete-time evolution operator under some effective Hamiltonian $H_{\rm eff}$ such that for $t \in \mathbb{N}$, 
\begin{equation}
    Q^t = \big(-(\mathbb{I}-2\ketbra{\psi}{\psi})(\mathbb{I}-2P)\big)^t = e^{-itH_{\rm eff}}.
\end{equation}
The following proposition from the main text and proof obtains $H_{\rm eff}$ and its eigenvalues.

\ham*
\begin{proof}
The state $\ket{\psi}$ can be decomposed into a good state and a bad state as follows:
\begin{equation}
    \ket{\psi} = \sin(\lambda)\ket{\psi_g} + \cos(\lambda)\ket{\psi_b},
\end{equation}
where
\begin{equation}
    \ket{\psi_g} = \frac{P\ket{\psi}}{\sqrt{\braket{\psi|P|\psi}}}, \quad
    \ket{\psi_b} = \frac{(\mathbb{I}-P)\ket{\psi}}{\sqrt{\braket{\psi|(\mathbb{I}-P)|\psi}}}.
\end{equation}

Amplitude amplification is implemented by the Grover iterate (up to a global phase)
\begin{equation}
    Q = -(\mathbb{I}-2\ketbra{\psi}{\psi})(\mathbb{I}-2P),
\end{equation}
where repeated application induces a coherent rotation in the two-dimensional subspace $\mathcal{H}_\psi = \mathrm{span}\{\ket{\psi_g},\ket{\psi_b}\}$~\citep{brassard2002quantum}.

Restricting to $\mathcal{H}_\psi$, the operator $Q$ acts as a real-valued rotation,
\begin{equation}
\begin{array}{c@{\;}c}
  & \begin{array}{cc}
      \scriptstyle\ket{\psi_g}\qquad & \,\scriptstyle\ket{\psi_b}
    \end{array} \\[0.4ex]
  \begin{array}{c}
    \scriptstyle\ket{\psi_g} \\
    \scriptstyle\ket{\psi_b}
  \end{array}
  &
  \left[
  \begin{array}{cc}
    \cos(2\lambda)  & \sin(2\lambda) \\
    -\sin(2\lambda) & \cos(2\lambda)
  \end{array}
  \right]
\end{array}
\end{equation}
which may be identified as a Pauli-Y rotation with angle $-2\lambda$. In the orthogonal complement space $\mathcal{H}_{\psi}^{\perp}$, we see that $\mathbb{I}-2\ketbra{\psi}{\psi}$ acts trivially, while $-(\mathbb{I}-2P)$ applies a phase $-1$ to the bad subspace $\mathcal{H}_B$. With this, we have constructed the spectrum for $Q$ in both $\mathcal{H}_{\psi}$ and $\mathcal{H}_{\psi}^\perp$. We can thus interpret $Q$ as the discrete-time propagator generated by an effective Hamiltonian $H_{\rm eff}$. Writing
\begin{equation}
    Y_{\psi}
    = i\ketbra{\psi_b}{\psi_g} - i\ketbra{\psi_g}{\psi_b},
\end{equation}
we may express $Q = e^{-iH_{\rm eff}}$ with
\begin{equation}
    H_{\rm eff} = -2\lambda\,Y_{\psi}
        + \pi\,(\mathbb{I}-P-\ketbra{\psi_b}{\psi_b}).
\end{equation}
The spectrum of $H_{\rm eff}$ consists of eigenvalues $\{\pm 2\lambda,\pi,0\}$, with the nontrivial dynamics confined to $\mathcal{H}_\psi$. Under this construction, we see that amplitude amplification corresponds to discrete-time evolution under the effective Hamiltonian $H_{\rm eff}$, and that $Q$ has eigenvalues $e^{\pm 2i\lambda}$ in subspace $\mathcal{H}_\psi$. Substituting $a = \sin^2(\lambda)$, we see that the eigenvalues are $e^{\pm 2i\sin^{-1}{\sqrt{a}}}$. 

To obtain the eigenvalue in the form of imaginary exponents of arccosines, we apply $Q$ to the $\ket{\psi}$ such that 
\begin{align}
Q\ket{\psi} &= -(\mathbb{I}-2\ketbra{\psi}{\psi})(\mathbb{I}-2P)\ket{\psi} \nonumber\\
&= -\ket{\psi} + 2 \ket{\psi} + 2P \ket{\psi} -4 \braket{\psi|P|\psi} \ket{\psi}\nonumber\\
&= (1-2a)\ket{\psi} + 2 (P\ket{\psi} - \braket{\psi|P|\psi}\ket{\psi})
\end{align}
where the latter term is orthogonal to $\ket{\psi}$. We then see that the eigenvalues of $Q$ are $e^{\pm i\cos^{-1}(1-2a)},$ corresponding to conventions in literature on qubitization.

To provide these results in the qubitization picture~\citep{low2019hamiltonian}, we can construct the reflection operator $R = U^\dagger (\mathbb{I} - 2P) U$. By applying qubitization over all qubits such that the entire Hilbert space is reduced to an effective 2-level system, we obtain the qubitized operator $S$ such that 
\begin{equation}
    S = R(2\ketbra{0}{0}-\mathbb{I}) = U^\dagger (\mathbb{I} - 2P) U(2\ketbra{0}{0}-\mathbb{I}).
\end{equation}
To obtain the eigenvalues of the qubitized values, we first note that
\begin{equation}
    \braket{0|U^\dagger (\mathbb{I} - 2P) U|0} = 1 -2 \braket{0|U^\dagger PU|0} = 1 - 2a.
\end{equation}
Thus, the eigenvalues of qubitized operator $S$ is then $e^{\pm i \cos^{-1}(1-2a)}$. Noting that $S = U^\dagger Q U$, we see that by unitary invariance of eigenvalues, the eigenvalues of $Q$ are also $e^{\pm i \cos^{-1}(1-2a)}$.
\end{proof}

\subsection{Proof of amplitude estimation as eigengap estimation}
\label{appLemEigap}

Given the construction of amplitude amplification as a discrete-time evolution operator, where the eigenvalues of the effective Hamiltonian in the nontrivial subspace $\mathcal{H}_{\psi}$ being $\pm E_{\rm eff} = \pm 2\lambda$, the eigengap in this subspace is thus $\Delta_{\rm eff} = E_{\rm eff} - (- E_{\rm eff}) = 4\lambda$. This can be extracted by eigengap estimation techniques, which forego the use of phase estimation, which requires at least an additional ancilla qubit and controlled time evolutions, aligning with amplitude estimation protocols without phase estimation~\citep{suzuki2020amplitude,aaronson2020quantum,grinko2021iterative,venkateswaran2021quantum,labib2024quantum}.
\eigap*
\begin{proof}
From the fact that $Q$ is a Pauli Y rotation in the subspace $\mathcal{H}_{\psi}$, we can see that $Q$ has two eigenvectors in the subspace $\mathcal{H}_{\psi}$ as follows:
\begin{equation}
    \ket{\psi_\pm} = \frac{1}{\sqrt{2}} (\ket{\psi_g} \pm i\ket{\psi_b})
\end{equation}
and with the eigenvalues of $e^{\pm iE_{\rm eff}}$, respectively. We now express $\ket{\psi}$ as the two eigenvectors as follows~\citep{brassard2002quantum}:
\begin{equation}
    \ket{\psi} = \frac{-i}{\sqrt{2}} \left(e^{i\lambda}\ket{\psi_+} - e^{-i\lambda}\ket{\psi_-}\right)
\end{equation}
We then see that
\begin{equation}
    Q^t\ket{\psi} = \frac{-i}{\sqrt{2}}(e^{i(tE_{\rm eff}+\lambda)}\ket{\psi_+} - e^{-i(tE_{\rm eff}+\lambda)}\ket{\psi_-})
\end{equation}
Upon projective measurement with $P$, we can then obtain
\begin{align}
    &\braket{\psi|Q^{-t}PQ^t|\psi} \nonumber\\
    &=
    \begin{aligned}[t]
    &\frac{1}{2} e^{i((tE_{\rm eff}+\lambda) -(tE_{\rm eff}+\lambda))}\braket{\psi_+|P|\psi_+}\\
    &- \frac{1}{2} e^{i((tE_{\rm eff}+\lambda)-(-(tE_{\rm eff}+\lambda)))}\braket{\psi_-|P|\psi_+}\\
    &- \frac{1}{2}e^{i((-(tE_{\rm eff}+\lambda))-(tE_{\rm eff}+\lambda))}\braket{\psi_+|P|\psi_-} \\
    &+\frac{1}{2}e^{i((-(tE_{\rm eff}+\lambda))-(-(tE_{\rm eff}+\lambda)))}\braket{\psi_-|P|\psi_-}
    \end{aligned}\nonumber\\
    &\begin{multlined}[b]=\frac{1}{2}\braket{\psi_+|P|\psi_+}- \frac{1}{2} e^{i(t\Delta_{\rm eff}+2\lambda)}\braket{\psi_-|P|\psi_+}\\
    -\frac{1}{2}e^{-i(t\Delta_{\rm eff}+2\lambda)}\braket{\psi_+|P|\psi_-} +\frac{1}{2}\braket{\psi_-|P|\psi_-}
    \end{multlined}
\end{align}
Then, calculating the four inner products, we obtain
\begin{multline}
    \braket{\psi_+|P|\psi_+} = \braket{\psi_-|P|\psi_+} = \braket{\psi_+|P|\psi_-} \\
    = \braket{\psi_-|P|\psi_-} =\frac{1}{2} \braket{\psi_g|P|\psi_g} =\frac{1}{2}.
\end{multline}
We thus obtain 
\begin{equation}
    \braket{\psi|Q^{-t}PQ^t|\psi} =\frac{1}{2}- \frac{1}{4} e^{i(t\Delta_{\rm eff}+2\lambda)} - \frac{1}{4}e^{-i(t\Delta_{\rm eff}+2\lambda)}
\end{equation}
Replacing the projective measurement with the observable $O = \mathbb{I}-2P$, we obtain 
\begin{equation}
    \braket{\psi|Q^{-t}OQ^t|\psi} = \frac{1}{2} e^{i(t\Delta_{\rm eff}+2\lambda)} + \frac{1}{2}e^{-i(t\Delta_{\rm eff}+2\lambda)}.
\end{equation}
Similarly, if we use a Loschmidt-echo-like projective measurement, we would obtain 
\begin{align}
    &\braket{\psi|Q^{-t}|\psi}\braket{\psi|Q^t|\psi} \nonumber\\
    &\begin{multlined}[b]=\frac{1}{2}\braket{\psi_+|\psi}\braket{\psi|\psi_+}- \frac{1}{2} e^{i(t\Delta_{\rm eff}+2\lambda)}\braket{\psi_-|\psi}\braket{\psi|\psi_+}\\
    - \frac{1}{2}e^{-i(t\Delta_{\rm eff}+2\lambda)}\braket{\psi_+|\psi}\braket{\psi|\psi_-} +\frac{1}{2}\braket{\psi_-|\psi}\braket{\psi|\psi_-}
    \end{multlined}
\end{align}
Calculating the inner products, we obtain 
\begin{align}
    \braket{\psi_+|\psi} = \frac{-i}{\sqrt{2}} e^{i\lambda},\quad \braket{\psi_-|\psi} = \frac{i}{\sqrt{2}} e^{-i\lambda}.
\end{align}
Plugging in the results, we obtain 
\begin{align}
    &\braket{\psi|Q^{-t}|\psi}\braket{\psi|Q^t|\psi} = \frac{1}{2} +\frac{1}{4} e^{it\Delta_{\rm eff}}
    + \frac{1}{4}e^{-it\Delta_{\rm eff}}.
\end{align}
Again replacing the projective measurement with the observable $O' = 2\ketbra{\psi}{\psi}-\mathbb{I}$, we obtain 
\begin{align}
    &\braket{\psi|Q^{-t}(2\ketbra{\psi}{\psi}-\mathbb{I})Q^t|\psi} = \frac{1}{2} e^{it\Delta_{\rm eff}}
    + \frac{1}{2}e^{-it\Delta_{\rm eff}}.
\end{align}
\end{proof}

\section{Properties of the periodic Gaussian}
\label{appGauss}
In this appendix, we establish several functional properties of the periodic Gaussian that are required for the analysis of our algorithms. In particular, we show that the periodic Gaussian satisfies strong convexity and smoothness conditions as shown in \cref{figGaussProp}. In our proofs in \cref{appThmMain}, we use the smoothness property to provide quadratic upper bounds and strong concavity to provide quadratic lower bounds for relevant regions.
\begin{figure}
    \includegraphics[width=\linewidth]{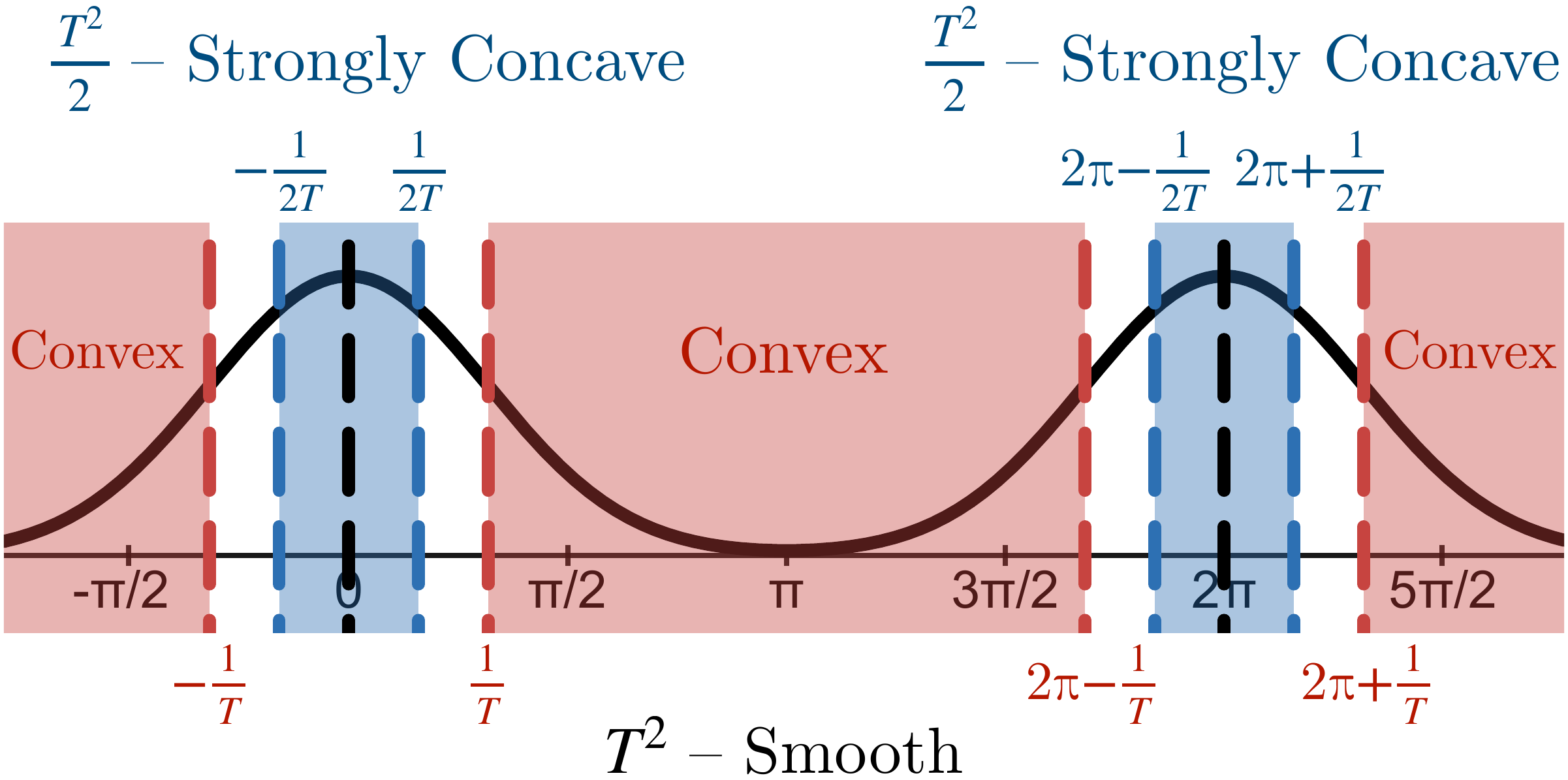}
    \caption{\emph{Regional functional properties of the periodic Gaussian.} In the periodic Gaussian with standard derivation $T\ge2/\pi$, up to periodicity $2\pi$, the region $[\frac{1}{T}, 2\pi -\frac{1}{T}]$ (red) is convex (\cref{lemCvxGauss}), the region $[-\frac{1}{2T}, \frac{1}{2T}]$ (blue) is $\frac{T^2}{2}$-strongly concave (\cref{lemStrongCvx}), and the entire function is $T^2$-smooth (\cref{lemSmoothGauss}).}
    \label{figGaussProp}
\end{figure}

We first define a simple Gaussian
\begin{equation}
\varphi_T(x) = \exp\left(-\frac{x^2T^2}{2}\right).
\end{equation}
Based on this Gaussian, we define a periodic Gaussian
\begin{equation}
\Phi_T(x) = \frac{1}{\mathcal{Z}}\sum_{j\in\mathbb{Z}}\varphi_T(x+2j\pi),
\end{equation}
where 
\begin{equation}
\mathcal{Z} = \sum_{j\in\mathbb{Z}} \varphi_T(2j\pi)=\sum_{j\in\mathbb{Z}}\exp\left(-\frac{(2j\pi)^2T^2}{2}\right),
\end{equation}
such that $\Phi_T(2j\pi) = 1$ for all $j\in\mathbb{Z}$.
Note that $\Phi_T(x)$ is the discrete-time Fourier/cosine transform of the discrete Gaussian
\begin{equation}
p_T(m) = \frac{1}{\sqrt{2\pi}T\mathcal{Z}}\exp\left(-\frac{m^2}{2T^2}\right)
\label{eqGaussPrior}
\end{equation}
where 
\begin{align}
    \Phi_T(k) = &\sum_{m=-\infty}^{\infty} p_T(m) \cos(km) \nonumber\\
    &= p_T(0) + 2\sum_{m=1}^{\infty} p_T(m) \cos(km)\nonumber\\
    &= \frac{1}{\mathcal{Z}}\sum_{j\in\mathbb{Z}}\exp\left(-\frac{(k+2j\pi)^2T^2}{2}\right).
\end{align} 
Thus, we can also find that 
\begin{equation}
\mathcal{Z} = \sum_{m=-\infty}^{\infty} \frac{1}{\sqrt{2\pi}T}\exp\left(-\frac{m^2}{2T^2}\right).
\end{equation}

We summarize the important lemmas used in future proofs as below (and in \cref{figGaussProp}):
\begin{restatable*}[Range of convexity of the periodic Gaussian]{lemma}{cvxGauss}
Let $T \ge 1/\pi$. Then $\Phi_T$ is convex on the interval $[\frac{1}{T}, 2\pi-\frac{1}{T}]$.
\label{lemCvxGauss}
\end{restatable*}
\begin{restatable*}[Smoothness of the periodic Gaussian]{lemma}{smooth}
Let $T\ge 2/\pi$, then $\Phi_T$ is $T^2$-smooth.
\label{lemSmoothGauss}
\end{restatable*}
\begin{restatable*}[Range of convexity of the negated periodic Gaussian]{lemma}{strongCvx}
\label{lemStrongCvx}
Let $T\ge 2/\pi$, then $1-\Phi_T$ is $\frac{T^2}{2}$-strongly convex on the interval $[-\frac{1}{2T}, \frac{1}{2T}]$.
\end{restatable*}
Readers interested only in the primary algorithmic contributions may skip forward to \cref{appThmMain}. The remainder of this appendix is included for completeness and for verification of the three lemmas above only.

\subsection{Local convexity}

We now provide some results on the convexity and smoothness of $\Phi_T$ for the ease of proof in later sections. Due to periodicity, results that specify a range are applicable with a periodicity of $2\pi$. For example, when we say convexity holds for and interval $[\frac{1}{T}, 2\pi - \frac{1}{T}]$, the results also hold for $[2j\pi + \frac{1}{T}, 2(j+1)\pi - \frac{1}{T}]$ for $j \in \mathbb{Z}$, To do so, we calculate the first four derivatives of $\varphi_T$ as follows:
\begin{align}
    \varphi_T'(x) &= -T^2x\varphi_T(x),\label{eqGauss1stDer}\\
    \varphi_T''(x) &= (T^4x^2 - T^2)\varphi_T(x),\label{eqGauss2ndDer}\\
    \varphi_T'''(x) &= -T^4x(T^2x^2-3)\varphi_T(x),\label{eqGauss3rdDer}\\
    \varphi_T''''(x) &= (T^8x^4 - 6T^6x^2+3T^4)\varphi_T(x).\label{eqGauss4thDer}
\end{align}

We first show a region of the periodic Gaussian that is convex, to both aid in the proof of later properties, and for the main proof.
\cvxGauss
\begin{proof}
From \cref{eqGauss2ndDer}, we have $\varphi_T''(x) = (T^4x^2 - T^2)\varphi_T(x)$. As $\varphi_T(x)>0$, we can see that $\varphi_T''(x) \ge 0$ if $x\in\mathbb{R}\setminus[-\frac{1}{T},\frac{1}{T}]$, given that $T^4x^2-T^2 \ge 0$. As derivatives are linear transforms, we note that
\begin{equation}
\Phi_T''(x) = \frac{1}{\mathcal{Z}}\sum_{j\in\mathbb{Z}}\varphi_T''(x+2j\pi).
\end{equation}
We can then see that if $x \in [\frac{1}{T}, 2\pi-\frac{1}{T}]$, then $\varphi_T''(x+2j\pi)\ge0$ for all $j \in \mathbb{Z}$, hence $\Phi_T''(x) \ge 0$, then $\Phi_T$ is convex. 
\end{proof}

\subsection{Smoothness}
We now show the smoothness of the periodic Gaussian to provide quadratic upper bounds on the function.
\begin{definition}[Smoothness of a function~\citep{nesterov2004introductory}]
Let $f : \mathbb R^d \to \mathbb R$, and $\beta>0$. We say that $f$ is $\beta$-smooth if it is differentiable
and if $\nabla f : \mathbb R^d \to \mathbb R^d$ is $\beta$-Lipschitz:
\begin{equation*}
\forall x,y \in \mathbb R^d
, \left\lVert\nabla f(y) -\nabla f(x)\right\rVert \le\beta \left\lVert y-x\right\rVert. 
\end{equation*}
\end{definition}
\begin{lemma}[Properties of $\beta$-smooth functions~\citep{nesterov2004introductory}]
\label{lemSmoothProp}
If $f$ is $\beta$-smooth and convex, then for all $x, y$:
\begin{equation*}
f(y) \le f(x) + \langle\nabla f(x), y-x\rangle + \frac{\beta}{2}\lVert y-x \rVert^2.
\end{equation*}
Further, let $\alpha \in [0, 1]$, then for all $x, y$
\begin{multline*}
\alpha f(x) + (1-\alpha) f(y) - f\left(\alpha x+ (1-\alpha) y\right) \\
\le \frac{\alpha(1-\alpha)\beta}{2}\lVert y-x \rVert^2
\end{multline*}
\end{lemma}

\smooth
\begin{proof}
From the definition of smoothness of a function and the Lipschitz property, we can find that the maximum value of the second derivative serves as a good $\beta$. Thus, the remainder of the proof is dedicated to bounding $\lvert\Phi_T''\rvert$. The high-level proof idea is to check critical points of $\Phi_T''$ and provide upper bounds on their magnitude.

To find local maxima/minima in $\Phi_T''$, we find the roots for the third derivative $\Phi_T'''$, indicating the existence of critical points, where again we can check by linearity of derivatives.
\begin{equation}
\Phi_T'''(x) = \frac{1}{\mathcal{Z}}\sum_{j\in\mathbb{Z}}\varphi_T'''(x+2j\pi).
\end{equation}
We first note that $\varphi_T'''$ is a odd function such that $\varphi_T'''(x)= -\varphi_T'''(-x)$, hence $x=0$ is a critical point. Pairing $\Phi_T'''(x+2j\pi)$ and $\Phi_T'''(x-2(j+1)\pi)$ where $j \in \mathbb{Z}_{\ge 0}$, such that $\Phi_T'''(x+2j\pi) = -\Phi_T'''(x-2(j+1)\pi)$, we find that $x=\pi$ is also a critical point. 

We also observe that from the third and fourth derivatives of $\varphi(x)$ in \cref{eqGauss3rdDer,eqGauss4thDer}, we see that $x=\frac{\sqrt{3}}{T}$ is a local maximum of $\varphi_T''(x)$ in $[\frac{1}{T},\infty)$. Plugging in $T \ge 2/\pi$, we can see that for all valid $T$, $\varphi(x)$ is monotonically decreasing for $x \in [\frac{\sqrt{3}}{2}\pi, \infty)$. In such ranges, we can use integrals to provide upper bounds.

\noindent \textit{Case 1: Critical point at $\Phi_T''(0)$.}

Recall that 
\begin{equation}
    \Phi_T''(0) = \frac{1}{\mathcal{Z}}\sum_{j\in\mathbb{Z}}\varphi_T''(2j\pi).
\end{equation}
Given that $T \ge 2/\pi$, for all $j \in \mathbb{Z}\setminus\{0\}$, $T^4(2\pi j)^2-T^2>0 $, thus, $\varphi_T''(2j\pi) > 0$, while $\varphi_T''(0) = -T^2$. 
We now show that 
\begin{equation}
    \sum_{j\in\mathbb{Z}\setminus\{0\}}\varphi_T''(2j\pi) < T^2.
\end{equation}
Given that we can factorize out $T^2$ from $\varphi_T''(x)$, we let $\varrho_T(x) = T^{-2}\varphi_T''(x)= (T^2x^2-1)\varphi_T(x)$.
We then see that
\begin{equation}
\sum_{j\in\mathbb{Z}\setminus\{0\}}\varphi_T''(2j\pi) = 2\sum_{j\in\mathbb{N}}\varphi_T''(2j\pi)= 2\sum_{j\in\mathbb{N}}T^2 \varrho_T(x).
\end{equation}

Given the monotonically decreasing properties of $\varphi_T''(2j\pi)$, and by extension, $\varrho_T(x)$, we can upper bound the tail with an integral and find
\begin{align}
&\sum_{j\in\mathbb{N}} \varrho_T(2j\pi) = \varrho_T(2\pi) + \sum_{j\in\mathbb{N}\setminus\{1\}} \varrho_T(2j\pi)\nonumber\\
&=(4\pi^2T^2-1)e^{-2\pi^2T^2} + \sum_{j\in\mathbb{N}\setminus\{1\}} (4\pi^2j^2T^2-1)e^{-2j^2\pi^2T^2}\nonumber\\
&\le(4\pi^2T^2-1)e^{-2\pi^2T^2} + \int_{j=1}^N (4\pi^2j^2T^2-1)e^{-2j^2\pi^2T^2}dj\nonumber\\
&=(4\pi^2T^2-1)e^{-2\pi^2T^2} + e^{-2\pi^2T^2}= 4\pi^2T^2e^{-2\pi^2T^2},
\end{align}
and
\begin{equation}
\sum_{j\in\mathbb{N}} \varrho_T(2j\pi) \ge \int_{j=1}^N (4\pi^2j^2T^2-1)e^{-2j^2\pi^2T^2}dj = e^{-2\pi^2T^2}.
\end{equation}
By the squeeze theorem, we find that 
\begin{equation}
\lim_{T\to \infty}\sum_{j\in\mathbb{N}} \varrho_T(2j\pi) = 0 
\end{equation}
and by extension,
\begin{equation}
\lim_{T\to \infty} \lvert\Phi_T(0)\rvert = \frac{T^2}{\mathcal{Z}} \le T^2.
\end{equation}

On the other end, we want to check if the bounds also apply for all $T\ge 2/\pi$. Taking the derivative of $4\pi^2T^2 e^{-2\pi^2T^2}$ over $T$, we find that the function is monotonically decreasing from $T = \frac{1}{\sqrt{2}\pi}$, so plugging in $T = 2/\pi$, we find
\begin{equation}
    \label{eqNatTail}
    \sum_{j\in\mathbb{N}}\varrho_T(2j\pi) \le \frac{16}{e^8} < 0.0055.
\end{equation}
Hence, we can find that 
\begin{equation}
\sum_{j\in\mathbb{Z}\setminus\{0\}}\varphi_T''(2j\pi) \le 0.0055T^2,
\end{equation}
and by extension 
\begin{align}
    \lvert\Phi_T''(0)\rvert &= \frac{1}{\mathcal{Z}} \left\lvert -T^2 + \sum_{j\in\mathbb{Z}\setminus\{0\}}\varphi_T''(2j\pi)\right\rvert \nonumber\\
    &\le \left\lvert -T^2 + \sum_{j\in\mathbb{Z}\setminus\{0\}}\varphi_T''(2j\pi)\right\rvert \le T^2
\end{align}

\noindent \textit{Case 2: Critical point at $\Phi_T''(\pi)$.}

Given that $x = \pi$ falls within $[\frac{1}{T},2\pi -\frac{1}{T}]$, as with the proof of \cref{lemCvxGauss}, we find that $\varphi_T''(x+2j\pi) \ge 0$ for all $j\in \mathbb{Z}$. Noting that $\Phi_T''((2j+1)\pi) =\Phi_T''(-(2j+1)\pi)$, we can find that 
\begin{equation}
\Phi_T''(\pi) = \frac{1}{\mathcal{Z}}\sum_{j\in\mathbb{Z}}\varphi_T''((2j+1)\pi) \le 2\sum_{j\in\mathbb{Z}_{\ge 0}}\varphi_T''((2j+1)\pi).
\end{equation}
Again letting $\varrho_T(x) = T^{-2}\varphi_T''(x)$, we can write the following:
\begin{align}
&\sum_{j\in\mathbb{Z}_{\ge 0}}\varrho_T''((2j+1)\pi) = \varrho_T(\pi) + \sum_{j\in\mathbb{N}}\varrho_{T}((2j+1)\pi)\nonumber\\
&\le (\pi^2T^2-1)e^{-\frac{\pi^2T^2}{2}} + \int_{j=0.5}^\infty(4\pi^2j^2T^2-1)e^{-2j^2\pi^2T^2}dj \nonumber\\
&= \frac{2\pi^2T^2-1}{2e^{0.5\pi^2T^2}}
\end{align}
Taking the derivative of the above over $T$, we find that the function is monotonically decreasing from $T = \sqrt{\frac{5}{2}}\frac{1}{\pi}$, so plugging in $T = 2/\pi$, we find
\begin{equation}
    \sum_{j\in\mathbb{Z}_{\ge 0}}\varrho_T((2j+1)\pi) \le \frac{7}{2e^2} < 0.475.
\end{equation}
Hence 
\begin{equation}
\Phi_T''(\pi) < 0.95 T^2 < T^2.
\end{equation}

\noindent \textit{Case 3: Additional critical points.}

Upon observing \cref{eqGauss3rdDer}, we note that there are 3 roots of $\varphi_T'''(x)$, $x = 0, \pm\frac{\sqrt{3}}{T}$. This hints that there is a potential third case of critical points around $x =  \frac{\sqrt{3}}{T}$ and $x = 2 \pi - \frac{\sqrt{3}}{T}$ for $\Phi_T''(x)$. From Case 1, we know that $\Phi''(0) < 0$, and therefore must be a global minimum. (And $\Phi''(2\pi) < 0$ for that matter.) Then, given that $x=\pi$ is also a critical point and would not be a saddle point due to symmetry, if $x=\pi$ is a local maximum, then the third set of critical points would be saddle points, and we need not worry. However, if $x=\pi$ is a local minimum, then this set of critical points would then be local maxima.

We focus on the range where a possible local maximum can appear in $x\in(0, \pi)$. We can obtain the following:
\begin{align}
&\Phi_T''(x) = \frac{1}{\mathcal{Z}}\sum_{j\in\mathbb{Z}}\varphi_T''(x+2j\pi)\nonumber\\
&\le \begin{multlined}[t]
    T^2(\varrho_T(x) + \varrho_T(x - 2\pi)) \\
    + T^2\sum_{j\in\mathbb{N}}(\varrho_T(x + 2j\pi)+\varrho_T(x - 2(j+1)\pi))
\end{multlined}
\end{align}
Factoring out the $T^2$, and from \cref{eqNatTail}, we can bound 
\begin{align}
&\begin{multlined}[t]
    (\varrho_T(x) + \varrho_T(x - 2\pi)) \\
    + \sum_{j\in\mathbb{N}}(\varrho_T(x + 2j\pi)+\varrho_T(x - 2(j+1)\pi))
\end{multlined}\nonumber\\
&=\begin{multlined}[t]
    (\varrho_T(x) + \varrho_T(2\pi-x)) \\
    + \sum_{j\in\mathbb{N}}(\varrho_T(x + 2j\pi)+\varrho_T(2(j+1)\pi-x))
\end{multlined}\nonumber\\
&\le 2\varrho_T\left(\frac{\sqrt{3}}{T}\right)+ \sum_{j\in\mathbb{N}}2\varrho_T(2j\pi)\nonumber\\
&\le \frac{4}{e^{1.5}} + \frac{32}{e^8} \le 0.905
\end{align}
Hence 
\begin{equation}
\Phi_T''(\pi) < 0.905 T^2 < T^2.
\end{equation}

\end{proof}

\subsection{Local strong concavity}
We now show the strong concavity in local regions of the periodic Gaussian to provide quadratic lower bounds of the function.
\begin{definition}[Strong convexity of a function~\citep{nesterov2004introductory}]
Let $f : \mathbb R^d \to \mathbb R$, and $\mu>0$. We say that $f$ is $\mu$-strongly convex if
\begin{equation*}
\forall x,y \in \mathbb R^d
, \left\langle\nabla f(y) -\nabla f(x), y-x\right\rangle \ge\mu \left\lVert y-x\right\rVert^2. 
\end{equation*}
\end{definition}
\begin{lemma}[Properties of $\mu$-strongly convex functions~\citep{nesterov2004introductory}]
\label{lemStrongCvxProp}
If and only if $f$ is $\mu$-strongly convex, for all $x$, 
\begin{equation*}
g(x) = f(x) -\frac{\mu}{2}\lVert x \rVert^2
\end{equation*}
is convex. 
For all $x, y$,
\begin{equation*}
f(y) \ge f(x) + \langle\nabla f(x), y-x\rangle + \frac{\mu}{2}\lVert y-x \rVert^2.
\end{equation*}
Further, let $\alpha \in [0, 1]$, then for all $x, y$
\begin{multline*}
\alpha f(x) + (1-\alpha) f(y) - f\left(\alpha x+ (1-\alpha) y\right) \\
\ge \frac{\alpha(1-\alpha)\mu}{2}\lVert y-x \rVert^2
\end{multline*}
\end{lemma}

\strongCvx
\begin{proof}
We want to show that on the interval $[-\chi, \chi]$, $1-\Phi_T$ is $\mu$-strongly convex, and find values for $\mu$ and $\chi$ under the condition that $T\ge 2/\pi$.

To show $\mu$ strong convexity for $1-\Phi_T$, we show that $1-\Phi_T(x) - \frac{\mu}{2}x^2$ is convex. Taking the second derivative, we can see that we want to find a range of $x$ and a value of $\mu$ such that
\begin{equation}
-\Phi_T''(x) - \mu > 0 \Rightarrow -\Phi_T''(x) > \mu > 0.
\end{equation}
Thus, we now attempt to find a lower bound for $-\Phi_T''(x)$ on the interval $[-\chi, \chi]$. 

From \cref{lemCvxGauss}, we know that for $x\in[\frac{1}{T}, 2\pi - T]$, $-\Phi_T''(x) \le 0$. Hence, we require $\chi < \frac{1}{T}$ to satisfy $-\Phi_T''(\chi) > 0$. For simplicity, we consider the case of $\chi = \frac{1}{2T}$.

We first show that $\Phi_T''(x)$ is monotonically increasing in the range $[0, \frac{1}{2T}]$, which we show by proving $\Phi_T'''(x)\ge 0$ in that region. By linearity, we once again see that 
\begin{align}
    &\Phi_T'''(x) = \frac{1}{\mathcal{Z}}\sum_{j\in\mathbb{Z}} \varphi_T'''(x + 2j\pi)\\
    &=\frac{1}{\mathcal{Z}}\left(\varphi_T'''(x)+\sum_{j\in\mathbb{N}} \left(\varphi_T'''(x + 2j\pi) +  \varphi_T'''(x - 2j\pi)\right)\right).
\end{align}
First we see that $\varphi_T'''(x) \ge 0$ when $x\in[0, \frac{1}{2T}]$ by solving the inequality $x(T^2x^2-3)\ge0$. Further, we note that from the fourth derivative $\varphi_T''''(x)$, we see that $\varphi_T'''(x)$ is positive and monotonically increasing in $\left(-\infty, -\frac{\sqrt{3+\sqrt{6}}}{T}\right]$ and thus $\left(-\infty, -\frac{\sqrt{3+\sqrt{6}}\pi}{2}\right]$. Thus, for all $j$, if $x\in[0, \frac{1}{2T}]$,
\begin{align}
    &\varphi_T'''(x + 2j\pi) +  \varphi_T'''(x - 2j\pi) \nonumber\\
    &= -\varphi_T'''(-x - 2j\pi) +  \varphi_T'''(x - 2j\pi)\nonumber\\
    &\ge -\varphi_T'''(-x - 2j\pi) +  \varphi_T'''(- 2j\pi)\ge 0.
\end{align}

We now show that for $-\Phi_T''(\frac{1}{2T}) > \frac{T^2}{2}$. Expanding the function, we obtain 
\begin{multline}
    -\Phi_T''\left(\frac{1}{2T}\right) = -\frac{1}{\mathcal{Z}}\varphi_T''\left(\frac{1}{2T}\right) - \frac{1}{\mathcal{Z}}\sum_{j\in\mathbb{N}}\varphi_T''\left(2j\pi + \frac{1}{2T}\right) \\
    - \frac{1}{\mathcal{Z}}\sum_{j\in\mathbb{N}}\varphi_T''\left(-2j\pi + \frac{1}{2T}\right)
\end{multline}
We first provide an upper bound for $\mathcal{Z}$. Recall that 
\begin{align}
\mathcal{Z} &= \sum_{j\in\mathbb{Z}}\varphi_T(2j\pi) = 1 + 2\varphi_T(2\pi) + 2\sum_{j\in{\mathbb{N}\setminus\{1\}}}\varphi_T(2\pi j)\nonumber\\
&\le 1 + 2e^{-2\pi^2T^2}+2\int_{j=1}^\infty e^{-2\pi^2j^2T^2} dj\nonumber\\
&< 1+ 2e^{-2\pi^2T^2} + \frac{1}{2\pi^2T^2}e^{-2\pi^2T^2}\nonumber\\
&\le1+ \frac{17}{8}e^{-8}<1.0008,
\end{align}
where we apply Mill's inequality~\citep{gordon1941values} for the second inequality. Then, we see that 
\begin{equation}
-\varphi_T''\left(\frac{1}{2T}\right)=\frac{3T^2}{4e^{0.125}}>0.66T^2
\end{equation}
Next, as we know that $\frac{1}{\mathcal{Z}}\sum_{j\in\mathbb{N}}\varphi_T''\left(2j\pi + \frac{1}{2T}\right)$ is positive, so we can upper bound this by
\begin{equation}
\sum_{j\in\mathbb{N}}\varphi_T''\left(2j\pi + \frac{1}{2T}\right) \le \sum_{j\in\mathbb{N}}\varphi_T''\left(2j\pi \right).
\end{equation}
We know by \cref{eqNatTail} that
\begin{equation}
    \sum_{j\in\mathbb{N}}\varphi_T''\left(2j\pi \right) = T^2 \sum_{j\in\mathbb{N}}\varrho_T(2j\pi) \le \frac{16T^2}{e^8} < 0.0055T^2.
\end{equation}
On the other side, we can bound
\begin{align}
&\sum_{j\in\mathbb{N}}\varphi_T''\left(-2j\pi + \frac{1}{2T}\right) = \sum_{j\in\mathbb{N}}\varphi_T''\left(2j\pi - \frac{1}{2T}\right) 
\nonumber\\
&\le \sum_{j\in\mathbb{N}}\varphi_T''\left((2j-0.25)\pi\right) \le T^2\sum_{j\in\mathbb{N}}\varrho_T''\left((2j-0.25)\pi\right)
\end{align}
Hence, we can then bound
\begin{align}
&\sum_{j\in\mathbb{N}}\varrho_T''\left((2j-0.25)\pi\right)\nonumber\\
&\le \varrho_T''\left(\frac{7}{4}\pi\right) + \int_{j=0.875}^\infty (4\pi^2j^2T^2-1)e^{-2j^2\pi^2T^2}dj\nonumber\\
&=\frac{49\pi^2T^2+13}{16}e^{-\frac{49\pi^{2} T^{2}}{32}}
\end{align}
Differentiating the above by $T$, we find that the function is monotonically decreasing from $T=\frac{\sqrt{19}}{7\pi} <\frac{2}{\pi}$. Hence, the maximum in our range can be found at $T=\frac{2}{\pi}$, and we bound
\begin{align}
\frac{49\pi^2T^2+13}{16}e^{-\frac{49\pi^{2} T^{2}}{32}}\le\frac{209}{16e^{6.125}} < 0.03
\end{align}
Thus, overall, we obtain
\begin{equation}
-\Phi_T''\left(\frac{1}{2T}\right) > \frac{0.66-0.0055-0.03}{1.0008}T^2 > 0.62T^2 > \frac{T^2}{2}
\end{equation}
\end{proof}

\section{Proofs of theoretical guarantees}
\label{appThmMain}

To find the appropriate sampling filter function, we require an even discrete function, as the estimates of the signal $S(m)$ are only accessible when $m$ is a positive integer, and a function that decays as $m$ grows, such that we can truncate the cases where $m$ is larger, and limit the length of the circuit. A suitable candidate, as noted by \citet{ding2024quantum}, is the discrete Gaussian established in \cref{appGauss} as follows:
\begin{equation}
    p_T(m) = \frac{1}{\sqrt{2\pi}T\mathcal{Z}}\exp\left(-\frac{m^2}{2T^2}\right)
\end{equation}
where $\mathcal{Z} = \sum_{m=-\infty}^{\infty} \frac{1}{\sqrt{2\pi}T}\exp\left(-\frac{m^2}{2T^2}\right)$. This serves as the probability distribution for sampling the evolution iteration $m$ for the amplitude amplification circuit. Recall that the discrete-time cosine transform~\citep{oppenheim2009discrete} of the $p_T$ is then 
\begin{align}
    \Phi_T(k)&= \sum_{m=-\infty}^{\infty} p_T(m) \cos(km)  \nonumber\\
    &= \frac{1}{\mathcal{Z}}\sum_{j\in\mathbb{Z}}\exp\left(-\frac{(k+2j\pi)^2T^2}{2}\right).
\end{align} 

In practice, we generate the discrete Gaussian by approximating $\mathcal{Z}$ with $\widetilde{\mathcal{Z}} = \sum_{m=-\rho T}^{\rho T} \frac{1}{\sqrt{2\pi}T}\exp\left(-\frac{m^2}{2T^2}\right)$ and apply a truncated version of the discrete Gaussian where 
\begin{equation}
    \widetilde{p}_T(m) = \begin{cases}
        \frac{1}{\sqrt{2\pi}T\widetilde{\mathcal{Z}}}\exp\left(-\frac{m^2}{2T^2}\right) & 1 \le |m| \le \sigma T\\
        0 & |m| > \sigma T\\
        1 - \sum_{m'\ne 0} \widetilde p_T(m') & m=0
    \end{cases}
    \label{eqTruncGauss}
\end{equation}
with $\sigma \le \rho$ to limit the maximum length at $M = \lfloor\sigma T\rfloor$ and approaching the Heisenberg limit. 

\begin{figure*}
    \includegraphics[width=0.9\linewidth]{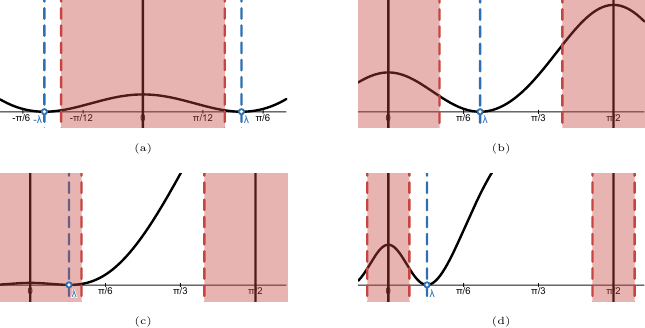}
    \caption{\emph{Failure regions for GLSAE.} We draft the ideal loss function for GLSAE without the measurement errors to show why GLSAE fails in regions close to 0 and $\pi/2$. In panel (a), we see that when plotted, the ideal loss function (black solid line) has symmetric minima at $\theta = \pm \lambda$ (blue dashed lines) formed from overlapping Gaussian peaks from the periodic Gaussians. In low depth situations, when the standard deviation $T$ of the discrete Gaussian is small, the minimum can still be found when $\lambda$ is far from either 0 or $\pi/2$, as shown in panel (b). However, when $\lambda$ is close to either 0 or $\pi/2$ (in the red shaded area), the minimum is not easy to discern, and worse so under noise. When the standard deviation $T$ of the discrete Gaussian is high enough, then the overlap can be reduced, and the invalid region (red shaded area) shrinks, such that we can still discern the peak even under noise, as known in panel (d).}
    \label{figGLSAEregion}
\end{figure*}

\begin{figure*}
    \includegraphics[width=0.90\linewidth]{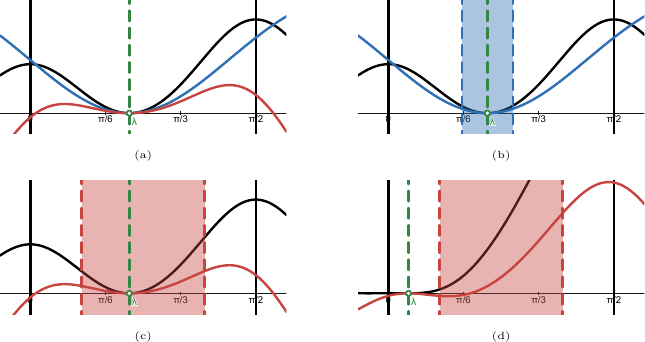}
    \caption{\emph{Analysis of the loss function failure region.} We illustrate the decomposition of the ideal loss function for GLSAE without the measurement errors to show how to quantify the regions where GLSAE fails. In panel (a), we decompose the ideal error function $\mathcal{E}(\theta)$ (black) as a sum of a negated periodic Gaussian $\mathcal{E}_1(\theta)$ (blue) and an indictor function $\mathcal{E}_2(\theta)$ (red) of the convexity of $\Phi_T$. In the blue shaded region of panel (b), strong convexity of $\mathcal{E}_1(\theta)$ can be found and used to provide quadratic lower bounds on $\mathcal{E}_1(\theta)$. In panel (c), the red shaded area indicates the range where $\Phi_T(4\theta)$ is convex, and hence the quadratic lower bounds hold for $\mathcal{E}(\theta)$ and the minimum is easily discernible. In panel (d), $\lambda$ is outside of the red shaded area where $\mathcal{E}_2(\theta)<0$, thus the quadratic lower bounds for $\mathcal{E}(\theta)$ may not hold. In this case, the minimum is not easily discernible, more so after noise is added.}
    \label{figGLSAEdecomp}
\end{figure*}

Using the truncated Gaussian, we define the loss function 
\begin{equation}
    \mathcal{L}_m(\theta) = \left(Z_m-\cos(2\theta m)\right)^2,
\end{equation}
where the total empirical loss can be defined as
\begin{equation}
    \widetilde{\mathcal{L}}(\theta) = \frac{1}{N}\sum_{m\sim \widetilde{p}_T}\left(Z_m-\cos(2\theta m)\right)^2.
\end{equation}
In an abuse of notation, we drop the subscript for $\mathcal{L}_m$ when $m$ is arbitrary, e.g., when $\mathcal{L}_m$ is wrapped in an expected value.
Calculating the expected value of the loss, we obtain
\begin{align}
    \mathbb{E}[\mathcal{L}(\theta)]
    &=  \mathbb{E}\left[\left(Z_m-\cos(2\theta m)\right)^2\right]\nonumber\\
    &=  \mathbb{E}[Z_m^2] - 2 \mathbb{E}[Z_m \cos(2\theta m)] + \mathbb{E}[\cos^2(2\theta m)]\nonumber\\
    &= 1 - 2 \mathbb{E}[\cos(2\lambda m) \cos(2\theta m)] + \mathbb{E}[\cos^2(2\theta m)]\nonumber\\
    &= \begin{multlined}[t]\frac{3}{2} - \mathbb{E}[\cos(2(\theta-\lambda)m)] - \mathbb{E}[\cos(2(\theta+\lambda)m)]\\
    + \frac{1}{2}\mathbb{E}[\cos(4\theta m)].
    \end{multlined}
\end{align}
where the second-to-last equality is obtained by observing  $\mathbb{E}[Z_m^2] = 1$ due to $Z_m = \pm 1$, as well as $\mathbb{E}[Z_m|m] = \cos(2\lambda m)$, and the last equality is obtained by expressing a product of cosines as a sum of cosines.

We now define an ideal error function $\mathcal{E}$ where the signal $\cos(2\lambda m)$ of a given time point is provided without sampling or noise, and no truncation is performed on the discrete Gaussian. In this case, the sampling probability takes the form of \cref{eqGaussPrior}. We then obtain the following:
\begin{align}
    &\mathcal{E}(\theta) = \sum_{m = -\infty}^{\infty} p_T(m) (\cos(2\lambda t)-\cos(2\theta m))^2\nonumber\\
    &= \begin{multlined}[t]
        \sum_{m = -\infty}^{\infty} p_T(m) \bigg(1+ \frac{1}{2}\cos(4\lambda m)- \cos(2(\theta-\lambda)m)\\
        - \cos(2(\theta+\lambda)m)
    + \frac{1}{2}\cos(4\theta m)\bigg)
    \end{multlined}\nonumber\\
    &= 1 + \frac{\Phi_T(4\lambda)}{2}  - \Phi_T(2(\theta-\lambda)) - \Phi_T(2(\theta+\lambda)) +\frac{\Phi_T(4\theta)}{2}
\end{align}
The main proof idea is to use the smoothness and convexity properties shown in \cref{appGauss} to provide upper and lower quadratic bounds. We can decompose $\mathcal{E}(\theta)$ into two terms 
\begin{align}
\mathcal{E}_1(\theta) &= 1- \Phi_T(2(\theta-\lambda))\nonumber\\
\mathcal{E}_2(\theta) &=\frac{\Phi_T(4\lambda)}{2}   - \Phi_T(2(\theta+\lambda)) +\frac{\Phi_T(4\theta)}{2},
\label{eqDecomp}
\end{align}
where we observe that $\mathcal{E}_1(\theta)$ is a negated periodic Gaussian centered at $2(\theta-\lambda)$ and $\mathcal{E}_2(\theta)$ is an indicator of the convexity of $\Phi_T$ in the range $[2\min(\theta,\lambda), 2\max(\theta,\lambda)]$. Note that the minimum of $\mathcal{E}_1(\theta)$ occurs when $\theta = \lambda$ up to periodicity $\pi$ of $\theta$ and is the solution we are looking for, and where we can apply smoothness and strong convexity arguments to provide quadratic upper and lower bounds.

However, since we can only access a noisy version of $\mathcal{E}(\theta)$ without isolating $\mathcal{E}_1(\theta)$, to ensure that the quadratic bounds of $\mathcal{E}_1(\theta)$ also hold for $\mathcal{E}(\theta)$, we should have $\mathcal{E}_2(\theta) \ge 0$, or where $\Phi_T$ is convex. Per \cref{lemCvxGauss}, we know that said range occurs when 
\begin{equation}
\lambda, \theta \in \left(\frac{1}{4T}, \frac{\pi}{2}-\frac{1}{4T}\right).
\end{equation}
Thus, when 
\begin{equation}
\lambda \lor \theta \in \left[0, \frac{1}{4T}\right) \cup \left(\frac{\pi}{2}-\frac{1}{4T}, \frac{\pi}{2}\right],
\end{equation}
$\mathcal{E}_2(\theta) <0$ and the quadratic bounds no longer hold for $\mathcal{E}(\theta)$. We illustrate this in \cref{figGLSAEdecomp}.

\subsection{Truncation and sampling errors}
\label{appProofBookkeep}
Before providing the main proof using the quadratic bounds as mentioned, we need to consider the effects of both the truncation error and the sampling error. We further define 
\begin{align}
\mathcal{F}(\theta) &= \mathcal{E}(\theta) + \frac{1}{2} - \frac{\Phi_T(4\lambda)}{2}\nonumber\\
&= \frac{3}{2} - \Phi_T(2(\theta-\lambda)) - \Phi_T(2(\theta+\lambda)) +\frac{\Phi_T(4\theta)}{2}
\end{align}
so that the terms correspond to $\mathbb{E}[\mathcal{L}(\theta)]$.

We now provide bounds on the error caused by the truncation of the discrete Gaussian $\widetilde{p}_T$ from \cref{eqTruncGauss} similar to Refs.~\citep{ding2024robust,ding2024quantum}.
\begin{lemma}[Truncation error of the loss from the discrete Gaussian]
\label{lemTruncErr}
Given $\delta > 0$, if $\sigma\in \Omega(\log^{\frac{1}{2}}(\frac{1}{\delta})))$ and $\rho\ge\sigma$, 
we have
\begin{equation}
\left\lvert\mathbb{E}\left[\mathcal{L}(\theta)\right] - \mathcal{F}(\theta)\right\rvert \le \delta^2
\end{equation}
\end{lemma}
\begin{proof}
We wish to find the truncation errors caused by two truncations: the first being the truncation error from truncating Gaussian samples, and the second from obtaining the normalization coefficient via a truncated sum.

From the triangle inequality, we can find that 
\begin{align}
\left\lvert\mathbb{E}\left[\mathcal{L}(\theta)\right] - \mathcal{F}(\theta)\right\rvert &\le \left\lvert\mathbb{E}[\cos(2(\theta-\lambda)m)] - \Phi(2(\theta-\lambda))\right\rvert \nonumber\\
&+ \left\lvert\mathbb{E}[\cos(2(\theta+\lambda)m)] - \Phi(2(\theta+\lambda))\right\rvert \nonumber\\
&+ \frac{1}{2}\left\lvert\mathbb{E}[\cos(4\theta m)] - \Phi(4\theta)\right\rvert
\end{align}
Given that the above 3 terms have the same structure, we assume a general $\gamma$ where we find
\begin{align}
&\operatorname*{\mathbb{E}}_{m\sim \widetilde{p}_T}\left[\cos(\gamma m)\right]
\nonumber\\
&=\sum_{|m|\le \sigma T} \widetilde p_T(m)\cos(\gamma m)\nonumber\\
&= \sum_{|m|\le \sigma T} \frac{\mathcal{Z}}{\widetilde{\mathcal{Z}}}p_T(m)\cos(\gamma m)+\left(1-\sum_{|m|\le \sigma T}\frac{\mathcal{Z}}{\widetilde{\mathcal{Z}}}p_T(m)\right)
\nonumber\\
& \begin{multlined}[b]
=\sum_{|m|\le \sigma T} \frac{\mathcal{Z}}{\widetilde{\mathcal{Z}}}p_T(m)\cos(\gamma m)\\\qquad+\left(\frac{\widetilde {\mathcal{Z}}-\mathcal{Z}}{\widetilde {\mathcal{Z}}}+\sum_{|m|> \sigma T}\frac{\mathcal{Z}}{\widetilde{\mathcal{Z}}}p_T(m)\right).
\end{multlined}
\end{align}
We can then observe that
\begin{align}
&\left\lvert\mathbb{E}\left[\cos(\gamma m)\right] - \Phi_T(\gamma)\right\rvert\nonumber\\
&=\left\lvert\sum_{|m|\le \sigma T} \widetilde p_T(m)\cos(\gamma m) - \sum_{m=-\infty}^\infty p_T(m)\cos(\gamma m)\right\rvert\nonumber\\
&\begin{multlined}\le \bigg\lvert\frac{\mathcal{Z}-\widetilde{\mathcal{Z}}}{\widetilde{\mathcal{Z}}}\sum_{m=-\infty}^\infty p_T(m)\cos(\gamma m) + \frac{\widetilde {\mathcal{Z}}-\mathcal{Z}}{\widetilde {\mathcal{Z}}}\\
+\sum_{|m|> \sigma T}\frac{\mathcal{Z}}{\widetilde{\mathcal{Z}}}p_T(m)(1-\cos(\gamma m)) \bigg\rvert
\end{multlined}\nonumber\\
&\le \left\lvert\frac{\mathcal{Z}-\widetilde{\mathcal{Z}}}{\widetilde{\mathcal{Z}}} (\Phi_T(\gamma)-1) \right\rvert
+\sum_{|m|> \sigma T}\left\lvert\frac{\mathcal{Z}}{\widetilde{\mathcal{Z}}}p_T(m)(1-\cos(\gamma m)) \right\rvert\nonumber\\
&\le \left\lvert\frac{\mathcal{Z}-\widetilde{\mathcal{Z}}}{\widetilde{\mathcal{Z}}} \right\rvert
+2\sum_{|m|> \sigma T}\frac{\mathcal{Z}}{\widetilde{\mathcal{Z}}}p_T(m).
\end{align}

Giving us the truncation error from truncating Gaussian samples as the second term, and the truncation in estimating the normalization coefficient as the first term.
The second term can be further upper-bounded by showing that
\begin{align}
\sum_{|m| > \sigma T} \frac{\mathcal{Z}}{\widetilde{\mathcal{Z}}}p_T(m) 
&= \sum_{|m| > \sigma T} \frac{1}{\sqrt{2\pi}T\widetilde{\mathcal{Z}}}e^{-\frac{m^2}{2T^2}}\nonumber\\
&\le \sum_{|m| > \sigma T} \frac{1}{\sqrt{2\pi}T}e^{-\frac{m^2}{2T^2}}\nonumber\\
&\le 2\int_{m = \sigma T}^{\infty}\frac{1sdxc}{\sqrt{2\pi}T}e^{-\frac{m^2}{2T^2}}dt\nonumber\\
&= \int_{s = \sigma }^{\infty}\sqrt{\frac{2}{\pi}}e^{-\frac{s^2}{2}}ds \le e^{-\sigma^2} \le \frac{2\delta^2}{15}
\end{align}
Similar to the above, the first term can be bounded as follows:
\begin{align}
\frac{\mathcal{Z}-\widetilde {\mathcal{Z}}}{\widetilde {\mathcal{Z}}} = \sum_{|m| > \rho T} \frac{1}{\sqrt{2\pi}T\widetilde{\mathcal{Z}}}e^{-\frac{m^2}{2T^2}}\le e^{-\rho^2} \le e^{-\sigma^2} \le \frac{2\delta^2}{15}.
\end{align}

Thus, for all $\gamma$, we can obtain 

\begin{equation}
\left\lvert\mathbb{E}\left[\cos(\gamma m)\right] - \Phi(\theta)\right\rvert \le\frac{2\delta^2}{5},
\end{equation}
and hence
\begin{equation}
\left\lvert\mathbb{E}\left[\mathcal{L}(\theta)\right] - \mathcal{F}(\theta)\right\rvert \le \delta^2.
\end{equation}
\end{proof}

We now obtain bounds regarding the number of samples $N$ via the Hoeffding bound of the difference of losses similar to Refs.~\citep{ding2024robust,ding2024quantum}. Let
\begin{equation}
    \mathcal{D}_m(\theta, \lambda) = \mathcal{L}_m(\theta)-\mathcal{L}_m(\lambda).
\end{equation}
We then obtain the following lemma.
\begin{lemma}[Sampling error from measurements]
\label{lemHoeffBound}
Let $0\le \xi \le 1$, $\delta \ge 0$, $\lambda \in [0, \frac{\pi}{2}]$ and $\kappa>0$. We define a grid $G$ in range $[0, \frac{\pi}{2}]$ with spacing $\kappa$ and let $J = \lfloor\frac{\pi}{2\kappa}\rfloor$ such that $G = \{\frac{k}{\kappa}:k \in [J]\cup\{0\}\}$. For all possible values of $\theta \in G$, if $N \in \mathcal{O}(\frac{1}{\delta^2} \log (\frac{1}{\kappa\xi}))$, then
\begin{equation*}
    \Pr\left[\left\lvert\frac{1}{N}\sum_{m}(\mathcal{D}_m(\theta, \lambda)) - \mathbb{E}[\mathcal{D}(\theta, \lambda)]\right\rvert\le\sigma T\delta\lvert\theta-\lambda\rvert\right] \ge 1-\xi.
\end{equation*}
\end{lemma}
\begin{proof}
Expanding the difference of losses, we obtain 
\begin{align}
    \mathcal{D}_m(\theta, \lambda) &= \mathcal{L}_m(\theta)-\mathcal{L}_m(\lambda)\nonumber\\
    &= \begin{multlined}[t]
        (1-2Z_m\cos(2\theta m)+\cos^2(2\theta m))\\
    - (1-2Z_m\cos(2\lambda m)+\cos^2(2\lambda m))
    \end{multlined}\nonumber\\
    &= \begin{multlined}[t]
        \cos^2(2\theta m) - \cos^2(2\lambda m) \\+ 2Z_m\cos(2\lambda m) - 2Z_m\cos(2\theta m)
    \end{multlined}\nonumber\\
    &\begin{multlined}[b]
    =(\cos(2\theta m) + \cos(2\lambda m) - 2Z_m)\\\times(\cos(2\theta m) - \cos(2\lambda m))
    \end{multlined}
\end{align}
Noting the inequality $\lvert\cos(2\theta m)) + \cos(2\lambda m) - 2Z_m\rvert\le 4$ and the Lipschitz property of cosines, we obtain the following:
\begin{equation}
    \lvert\mathcal{D}_m(\theta, \lambda)\rvert \le 4 (\cos(2\lambda m) - \cos(2\theta m)) \le 8 \sigma T \lvert\theta-\lambda\rvert
\end{equation}
where the last inequality is obtained by noting that for all $m$, $m \le \sigma T$. From this bound, we note that $\mathcal{D}_m(\theta, \lambda) \in [-8 \sigma T \lvert\theta-\lambda\rvert, 8 \sigma T \lvert\theta-\lambda\rvert]$, which we can use as the bounds for Hoeffding's inequality, where we can obtain
\begin{multline}
    \Pr\left[\left\lvert\frac{1}{N}\sum_{m}(\mathcal{D}_m(\theta, \lambda)) - \mathbb{E}[\mathcal{D}(\theta, \lambda)]\right\rvert\ge \eta \right] \\
    \le \exp\left(-\frac{2N\eta^2}{(16\sigma T \lvert\theta - \lambda|)^2}\right).
\end{multline}
We note that $|G| = J +1 = \lfloor\frac{\pi}{2\kappa}\rfloor +1$, thus by union bound for $\theta \in G$,
\begin{multline}
    \Pr\left[\bigcup_{\theta \in G }\left\{\left\lvert\frac{1}{N}\sum_{m}(\mathcal{D}_m(\theta, \lambda)) - \mathbb{E}[\mathcal{D}(\theta, \lambda)]\right\rvert\ge \eta \right\}\right] \\
    \le (J+1)\exp\left(-\frac{2N\eta^2}{(16\sigma T \lvert\theta - \lambda|)^2}\right).
\end{multline}
Upper-bounding the entire thing with $\xi$, we obtain 
\begin{equation}
\eta  \ge \frac{16\sigma T\lvert\theta - \lambda|}{\sqrt{2N}}\ln^{1/2}\left(\frac{J+1}{\xi}\right).
\end{equation}
Let $\delta = \frac{16}{\sqrt{2N}}\ln^{1/2}\left(\frac{J+1}{\xi}\right)$, then rearranging the terms, we see that
\begin{equation}
N \in \mathcal{O}\left(\frac{1}{\delta^2}\ln\left(\frac{1}{\kappa\xi}\right)\right),
\end{equation}
and that $\eta = \sigma T\delta\lvert\theta-\lambda\rvert$.
\end{proof}

\subsection{Proof of GLSAE guarantees}
\label{appProofMain}
We can now prove our main result. The high level proof idea is to provide upper bounds for $\mathcal{L}(\theta)$ when $\lvert\theta-\lambda\rvert \le \kappa$ and lower bounds for $\widetilde{\mathcal{L}}(\theta)$ when $\lvert\theta-\lambda\rvert \ge 2\kappa$ to ensure that minimizing $\mathcal{L}(\theta)$ provides a close enough estimate of $\lambda$. To do this, we connect $\widetilde{\mathcal{L}}(\theta)$ to the ideal error function $\mathcal{E}(\theta)$ via bounding the sampling error and truncation errors derived in the previous section. Lastly, we provide quadratic upper and lower bounds on $\mathcal{E}(\theta)$ when applicable to provide the $M^2N \in \widetilde{\mathcal{O}}(\epsilon^{-2})$ invariance by setting $\epsilon = 2\kappa$.
\main*
\begin{proof}
    Recall that $M = \lfloor \sigma T\rfloor$. For ease of proof, we view the algorithm through the lens of $T$ instead of $M$. Further recall that the values of $\theta$ are obtained via grid search with spacing $\kappa$.
    
    We first show that we can construct cases where $\max_{\lvert\theta-\lambda\rvert \le \kappa}\mathcal{L}(\theta) \le \min_{\lvert\theta-\lambda\rvert \ge 2\kappa} \mathcal{L}(\theta)$.

    From \cref{lemHoeffBound}, we can see that with probability greater than $1-\xi$,
    \begin{equation}
        \left\lvert\left(\widetilde{\mathcal{L}}(\theta) - \widetilde{\mathcal{L}}(\lambda)\right) - \left(\mathbb{E}[\mathcal{L}(\theta)] - \mathbb{E}[\mathcal{L}(\lambda)]\right)\right\rvert \le \sigma T \delta \lvert\theta-\lambda\rvert
    \end{equation}
    Rearranging the terms, we find that 
    \begin{align}
        &\widetilde{\mathcal{L}}(\theta) \le \widetilde{\mathcal{L}}(\lambda) + \mathbb{E}[\mathcal{L}(\theta)] - \mathbb{E}[\mathcal{L}(\lambda)] + \sigma T \delta \lvert\theta-\lambda\rvert\nonumber\\
        &\widetilde{\mathcal{L}}(\theta) \ge \widetilde{\mathcal{L}}(\lambda) + \mathbb{E}[\mathcal{L}(\theta)] - \mathbb{E}[\mathcal{L}(\lambda)] - \sigma T \delta \lvert\theta-\lambda\rvert
    \end{align} 
    From \cref{lemTruncErr}, we obtain
    \begin{equation}
    -2\delta'^2 \le (\mathbb{E}[\mathcal{L}(\theta)]-\mathbb{E}[\mathcal{L}(\lambda)]) - (\mathcal{F}(\theta)-\mathcal{F}(\lambda)) \le 2\delta'^2.
    \end{equation}
    Let $\delta' = \delta$ such that $N \in \mathcal{O}(\frac{1}{\delta^2} \log (\frac{1}{\kappa\xi}))$ and $\sigma\in \Omega(\log^{\frac{1}{2}}(\frac{1}{\delta})))$.
    From the definition of $\mathcal{F}$, we note that $\mathcal{F}(\theta)-\mathcal{F}(\lambda) = \mathcal{E}(\theta)-\mathcal{E}(\lambda) = \mathcal{E}(\theta)$. Hence, we can write 
    \begin{align}
        \label{eqLossHoff}
        \widetilde{\mathcal{L}}(\theta) \le \widetilde{\mathcal{L}}(\lambda) + \mathcal{E}(\theta) + \sigma T \delta \lvert\theta-\lambda\rvert +2\delta^2\nonumber\\
        \widetilde{\mathcal{L}}(\theta) \ge \widetilde{\mathcal{L}}(\lambda) + \mathcal{E}(\theta) - \sigma T \delta \lvert\theta-\lambda\rvert -2\delta^2
    \end{align}
    
    We now discuss the following two cases:
    
    \noindent\textsc{Case 1: $\lvert\theta-\lambda\rvert \le \kappa$.}
    
    In this case, we provide an upper bound for $\mathcal{E}(\theta)$ by showing smoothness bounds for both $\mathcal{E}_1(\theta)$ and $\mathcal{E}_2(\theta)$.
    When $T\ge 2/\pi$, by smoothness properties in \cref{lemSmoothProp,lemSmoothGauss},
    \begin{align}
        \mathcal{E}_1(\theta)  &= 1 - \Phi_T(2(\theta-\lambda)) \nonumber\\
        &\le 1 - \Phi_T(0) - \Phi_T'(0)\times(2(\theta-\lambda)) +\frac{T^2}{2} (2(\theta-\lambda))^2 \nonumber\\
        &= \frac{T^2}{2} (2(\theta-\lambda))^2 = 2 T^2(\theta-\lambda)^2
    \end{align}
    Further, again by smoothness properties in \cref{lemSmoothProp,lemSmoothGauss},
    \begin{align}
        \mathcal{E}_2(\theta) &=\frac{\Phi_T(4\lambda)}{2} - \Phi_T(2(\theta+\lambda)) +\frac{\Phi_T(4\theta)}{2}\nonumber\\
        &\le \frac{T^2}{8} (2(\theta-\lambda))^2 = \frac{1}{2} T^2(\theta-\lambda)^2
    \end{align}
    Combining the results with the sampling and truncation error bounds in \cref{eqLossHoff}, we obtain an upper bound on $\widetilde{\mathcal{L}}(\theta)$ such that
    \begin{align}
        \widetilde{\mathcal{L}}(\theta) &\le \widetilde{\mathcal{L}}(\lambda)+\frac{5}{2}T^2(\theta-\lambda)^2 + \sigma T \delta |\theta-\lambda| + 2\delta^2\nonumber\\
        &\le \widetilde{\mathcal{L}}(\lambda)+\frac{5}{2}T^2\kappa^2 + \sigma T \delta \kappa + 2\delta^2
    \end{align}

    \noindent\textsc{Case 2: $\lvert\theta-\lambda\rvert \ge 2\kappa$.}
    
    In this case, we provide a lower bound for $\mathcal{E}(\theta)$.
    Recall that by convexity from \cref{lemCvxGauss}, when $x \in [\frac{1}{T}, 2\pi-\frac{1}{T}]$, $\Phi_T(x)$ is convex. Thus when both $\theta$ and $\lambda$ are in the range $[\frac{1}{4T}, \frac{\pi}{2}-\frac{1}{4T}]$, we have
    \begin{equation}
        \mathcal{E}_2(\theta) = \frac{\Phi_T(4\lambda)}{2}  - \Phi_T(2(\theta+\lambda)) +\frac{\Phi_T(4\theta)}{2} \ge 0.
    \end{equation}
    Further, when $T\ge 2/\pi$ and if $2\kappa \le \frac{1}{2T}$, by strong convexity from \cref{lemStrongCvx},
    \begin{align}
        \mathcal{E}_1(\theta) &= 1 - \Phi_T(2(\theta-\lambda)) \nonumber\\
        &\ge 1 - \Phi_T(0) - \Phi_T'(0)\times2(\theta-\lambda) + \frac{T^2}{4}(2(\theta-\lambda))^2\nonumber\\
        &= \frac{T^2}{4}(2(\theta-\lambda))^2 \ge T^2(\theta-\lambda)^2
    \end{align}
    Combining the results with sampling and truncation error bounds in \cref{eqLossHoff}, we obtain an lower bound on $\widetilde{\mathcal{L}}(\theta)$ such that if $2\kappa \ge \frac{\sigma\delta}{T}$, then
    \begin{align}
        \widetilde{\mathcal{L}}(\theta) &\ge \widetilde{\mathcal{L}}(\lambda)+T^2(\theta-\lambda)^2 - \sigma T \delta |\theta-\lambda| - 2\delta^2\nonumber\\
        &\ge \widetilde{\mathcal{L}}(\lambda)+4T^2\kappa^2 - 2\sigma T \delta \kappa - 2\delta^2.
    \end{align}

    We can then see that
    \begin{align}
        \max_{\lvert\theta-\lambda\rvert \le \kappa}\widetilde{\mathcal{L}}(\theta) &\le \widetilde{\mathcal{L}}(\lambda)+\frac{5}{2}T^2\kappa^2 + \sigma T \delta\kappa + 2\delta^2\nonumber\\
        \min_{\lvert\theta-\lambda\rvert \ge 2\kappa} \widetilde{\mathcal{L}}(\theta) &\ge \widetilde{\mathcal{L}}(\lambda)+4T^2\kappa^2-2\sigma T \delta \kappa - 2\delta^2.
    \end{align}
    Solving for
    \begin{equation}
    \frac{5}{2}T^2\kappa^2 + \sigma T \delta\kappa + 2\delta^2\le 4T^2\kappa^2-2\sigma T \delta \kappa - 2\delta^2,
    \end{equation}
    we obtain the inequality
    \begin{equation}
    \frac{3}{2}T^2\kappa^2-3\sigma T \delta \kappa - 4\delta^2 \ge 0,
    \end{equation}
    and 
    \begin{equation}
    T \ge \frac{3\sigma \delta\kappa +\delta\kappa\sqrt{9\sigma^2+24}}{3\kappa^2} = \frac{\delta}{\kappa} \frac{3\sigma +\sqrt{9\sigma^2+24}}{3},
    \end{equation}
    in order to fulfill $\max_{\lvert\theta-\lambda\rvert \le \kappa}\widetilde{\mathcal{L}}(\theta) \le \min_{\lvert\theta-\lambda\rvert \ge 2\kappa} \widetilde{\mathcal{L}}(\theta)$.

    Rearranging the terms and plugging in the relevant factors, we find
    \begin{equation}
        T \in \mathcal{O}\left(\frac{1}{\sqrt{N}\kappa}\log^{\frac{1}{2}}\left(\frac{1}{\kappa\xi}\right)\log^{\frac{1}{2}}\left(\frac{N}{\log(\kappa^{-1}\xi^{-1})
        }\right)\right).
    \end{equation}
    
    The above setting of $T$ ensures that for all $|\theta-\lambda| \ge 2\kappa$, the loss function is larger than that of all $|\theta-\lambda| \le \kappa$. However, there still exists an annulus where $\kappa \le |\theta-\lambda| \le 2\kappa$, where the behavior is undefined. By setting $\epsilon = 2\kappa$, we let 
    \begin{equation}
        \theta^* = \argmin_{\theta \in G} \widetilde{\mathcal L}(\theta)
    \end{equation}
    then per inequality above, we know that $\theta^*$ must lie within $\epsilon = 2\kappa$ of $\lambda$.

    To recover the asymptotics of $M$ and $N$ in terms of $\epsilon$ as in the main theorem statement, we plug in $M= \lfloor \sigma T\rfloor$ and $\epsilon = 2\kappa$ to obtain 
    \begin{equation}
        \label{eqMaxDepth}
        M \in \mathcal{O}\left(\frac{1}{\sqrt{N}\epsilon}\log^{\frac{1}{2}}\left(\frac{1}{\xi\epsilon}\right)\log\left(\frac{N}{\log(\xi^{-1}\epsilon^{-1})
        }\right)\right).
    \end{equation}
    We fix a tuning parameter $\beta \in [0,1]$ such that $N\in\widetilde{\mathcal{O}}(\epsilon^{-2\beta})$. We can then find that $M\in\widetilde{\mathcal{O}}(\epsilon^{-1+\beta})$.
    
    Lastly, we note that as the $\epsilon$ is the bound for the phase difference $\lvert \theta - \lambda\rvert$, we still need to relate it to the amplitude $\widetilde a = \sin^2(\theta)$. From Lemma 7 of \citet{brassard2002quantum}, we see that 
    \begin{equation}
        \lvert\widetilde{a} - a| \le 2\sqrt{a(1-a)} |\theta-\lambda| + |\theta-\lambda|^2 \le 2\sqrt{a(1-a)} \epsilon + \epsilon^2
    \end{equation}

    From the valid range of $\theta$ and $\lambda$ from maintaining $\mathcal{E}_2(\theta)\ge 0$, we require $\theta, \lambda \in [\frac{1}{4T}, \frac{\pi}{2} - \frac{1}{4T}]$ so the valid range for estimating $a$ is then $[\sin^2(\frac{1}{4T}), \cos^2(\frac{1}{4T})]$. Setting $\zeta = \sin^2(\frac{1}{4T})$, we find that
    \begin{equation}
        \sin^2\left(\frac{1}{4T}\right) \le \frac{1}{16T^2} \in \widetilde{\mathcal{O}}(\epsilon^{2-2\beta}).
    \end{equation}
\end{proof}

We further provide specific results regarding the runtime at the Heisenberg limit range where $M\in \mathcal{O}(\epsilon^{-1})$.
\corHeis*
\begin{proof}
    Let $N \in \mathcal{O}(\log(\xi^{-1}\epsilon^{-1}))$. \cref{eqMaxDepth} produces $M \in \mathcal{O}(\epsilon^{-1})$. In this regime, we note that $\zeta \in \widetilde{\mathcal O} (\epsilon^2)$, and thus, we have $\zeta \ge \epsilon$. As a result, even when $a \in [0, \zeta)$, the estimate $\widetilde{a}$ found would be either $0$ or $\epsilon$, both of which are valid outputs. The same follows for $a \in (1-\zeta, 1]$. Thus, when $\beta = 0$, the range of estimation for $a$ is over the entire $[0, 1]$.
\end{proof}

\subsection{Acceptance certification}
\label{appProofCert}
Since \cref{thmMain} assumes that the target amplitude $a$ lies in the interval $[\zeta, 1-\zeta]$, applying GLSAE reliably requires prior knowledge that $a$ satisfies this condition. In practice, such prior information may not be available, making it unclear whether a given GLSAE output should be trusted.

To address this limitation, we introduce in this appendix an acceptance criterion based solely on the observed GLSAE output. This criterion allows one to decide directly from the readout values whether to accept or reject the estimate.

\begin{table}[h]
\centering
\begin{tabular}{ccc}
\toprule
 & $\theta \in \left[\frac{1}{4T}, \frac{\pi}{2}-\frac{1}{4T}\right]$ & $\theta \notin \left[\frac{1}{4T}, \frac{\pi}{2}-\frac{1}{4T}\right]$ \\
\midrule
$\lambda \in \left[\frac{1}{4T}, \frac{\pi}{2}-\frac{1}{4T}\right]$ & Accept & Reject \\ 
$\lambda \notin \left[\frac{1}{4T}, \frac{\pi}{2}-\frac{1}{4T}\right]$ & Reject & Reject \\
\bottomrule
\end{tabular}
\end{table}

From the above table, it follows that any estimate $\theta$ falling outside the interval $\left[\frac{1}{4T}, \frac{\pi}{2}-\frac{1}{4T}\right]$ can be safely rejected. However, if $\theta$ lies within this interval, the situation is more ambiguous, as it may correspond either to the valid regime $\lambda \in \left[\frac{1}{4T}, \frac{\pi}{2}-\frac{1}{4T}\right]$ or to the invalid regime $\lambda \notin \left[\frac{1}{4T}, \frac{\pi}{2}-\frac{1}{4T}\right]$, in which case the output should still be rejected.

Therefore, we require a stricter rejection criterion that eliminates all invalid cases within this regime, even at the expense of occasionally rejecting valid estimates. This can be achieved by enlarging the rejection region, or equivalently, by defining a more restrictive acceptance region $\mathcal{A}$.

Specifically, we seek an acceptance region $\mathcal{A}$ such that, whenever $\lambda < \frac{1}{4T}$, the objective function satisfies
\begin{equation}
    \min_{\theta \in \mathcal{A}} \widetilde{\mathcal{L}}(\theta) > \max_{\theta < \frac{1}{T}} \widetilde{\mathcal{L}}(\theta),
\end{equation}
as well as for $\lambda > \frac{\pi}{2} - \frac{1}{4T}$.

This condition ensures that, during the minimization step of GLSAE, any estimate $\theta$ lying within the acceptance region $\mathcal{A}$ cannot correspond to a true amplitude $a \notin (\zeta, 1-\zeta)$. In the following, we show that the range
\begin{equation}
    \mathcal{A} = \left[\frac{1}{T}, \frac{\pi}{2}-\frac{1}{T}\right],
\end{equation}
and, by extension, the range of $[16\zeta, 1-16\zeta]$ for amplitudes serves as a good candidate for our purposes.

\certi*
\begin{proof}
    We define $\mathcal{A}$ to be the range $\left[\frac{1}{T}, \frac{\pi}{2}-\frac{1}{T}\right]$. To show a upper bound for the case where $\theta < \frac{1}{4T}$, we start from the Hoeffding and truncation bounds from \cref{eqLossHoff}, we obtain
    \begin{align}
        \widetilde{\mathcal{L}}(\theta) &\le \widetilde{\mathcal{L}}(\lambda) + \mathcal{E}(\theta) + \sigma T \delta \lvert\theta-\lambda\rvert +2\delta^2\nonumber\\
        &\le \widetilde{\mathcal{L}}(\lambda) + \mathcal{E}_1(\theta) + \mathcal{E}_2(\theta) + \frac{\sigma\delta}{4}  +2\delta^2
    \end{align}
    We continue by upper-bounding $\mathcal{E}_2(\theta)$ in this regime. Given that $\lambda,\theta < \frac{1}{4T}$, the continuous Gaussian $\varphi_T(x)$ is convex so 
    \begin{equation}
    \frac{\varphi_T(4\lambda)}{2}  - \varphi_T(2(\theta+\lambda)) +\frac{\varphi_T(4\theta)}{2} < 0
    \end{equation}

    Recall from \cref{eqDecomp} that
    \begin{align}
        \mathcal{E}_2(\theta) &=\frac{\Phi_T(4\lambda)}{2}   - \Phi_T(2(\theta+\lambda)) +\frac{\Phi_T(4\theta)}{2}\nonumber\\
        &=\begin{multlined}[t][0.85\linewidth]\frac{1}{\mathcal{Z}}\sum_{j\in\mathbb{Z}}\bigg(\frac{\varphi_T(4\lambda + 2j\pi)}{2}  \\- \varphi_T(2(\theta+\lambda)+ 2j\pi) +\frac{\varphi_T(4\theta+ 2j\pi)}{2}\bigg)
        \end{multlined}\nonumber\\
        &\le \sum_{j\in\mathbb{Z}\setminus\{0\}}\frac{\varphi_T(4\lambda + 2j\pi)}{2} +\frac{\varphi_T(4\theta+ 2j\pi)}{2}\nonumber\\
        &\le \sum_{j\in\mathbb{N}}2\varphi_T\left(2j\pi \right) + 2\varphi_T\left(2j\pi - \frac{1}{T} \right),
    \end{align}
    where we use the fact that $\mathcal{Z}>0$ and $\varphi_T$ is concave in the first inequality. Under the condition of $T \ge \frac{4}{\pi}$, we have 
    \begin{align}
        \mathcal{E}_2(\theta) &\le \sum_{j\in\mathbb{N}}2\varphi_T\left(2j\pi \right) + 2\varphi_T\left(2j\pi-\frac{\pi}{4}\right)\nonumber\\
        &\begin{multlined}[t][0.85\linewidth]\le 2e^{-2\pi^2T^2} + 2e^{-\frac{49}{32}\pi^2T^2}\\
            + \int_{j=1}^\infty2\varphi_T\left(2j\pi \right) + \int_{j=7/8}^\infty 2\varphi_T\left(2j\pi\right)
        \end{multlined}\nonumber\\
        &\begin{multlined}[t][0.85\linewidth]\le2e^{-2\pi^2T^2} + 2e^{-\frac{49}{32}\pi^2T^2}\\
        + \frac{1}{2\pi^2 T^2} e^{-2\pi^2T^2} + \frac{4}{7\pi^2 T^2} e^{-\frac{49}{32}\pi^2T^2}
        \end{multlined}\nonumber\\
        &\le\frac{33}{32}e^{-32} + \frac{57}{28}e^{-\frac{49}{2}}<5\times 10^{-11}.
    \end{align}
    Thus, we can write 
    \begin{equation}
        \max_{\theta<\frac{1}{4T}} \widetilde{\mathcal{L}}(\theta) \le \widetilde{\mathcal{L}}(\lambda) + \max_{\theta<\frac{1}{4T}}\mathcal{E}_1(\theta) + \frac{\sigma\delta}{4}  +2\delta^2 + 5\times 10^{-11}.
    \end{equation}

    To show a lower bound for the case where $\theta \in \mathcal{A}$, we assume the lower threshold for $\mathcal{A}$ is $\frac{1}{T}$. We again start from \cref{eqLossHoff}, but this time we tweak the range for $\lvert\mathcal{D}(\theta, \lambda)\rvert$ such that it is not dependent on $\lvert\theta-\lambda\rvert$. This resulting change in Hoeffding's inequality effectively produces
    \begin{align}
        \widetilde{\mathcal{L}}(\theta) &\ge \widetilde{\mathcal{L}}(\lambda) + \mathcal{E}_1(\theta) + \mathcal{E}_2(\theta) - \sigma \delta -2\delta^2.
    \end{align}
    We again continue lower bounding $\mathcal{E}_2(\theta)$ where
    \begin{align}
        \mathcal{E}_2(\theta) &\begin{multlined}[t][0.9\linewidth]=\frac{1}{\mathcal{Z}}\sum_{j\in\mathbb{Z}}\bigg(\frac{\varphi_T(4\lambda + 2j\pi)}{2}  \\- \varphi_T(2(\theta+\lambda)+ 2j\pi) +\frac{\varphi_T(4\theta+ 2j\pi)}{2}\bigg)
        \end{multlined}\nonumber\\
        &\begin{multlined}[b][0.9\linewidth]\ge\frac{1}{\mathcal{Z}} \bigg(\frac{\varphi_T(4\lambda)}{2}   - \varphi_T(2(\theta+\lambda)) +\frac{\varphi_T(4\theta)}{2} \\- \sum_{j\in\mathbb{Z}\setminus\{0\}}\varphi_T(2(\theta+\lambda)+2j\pi)\bigg)\end{multlined}\nonumber\\
        &\begin{multlined}[b][0.9\linewidth]\ge\frac{1}{\mathcal{Z}} \bigg(\frac{1}{2}\varphi_T\left(\frac{1}{T}\right)   - \varphi_T\left(\frac{2}{T}\right)  + 0\\- \sum_{j\in\mathbb{N}}\left(2\varphi_T\left(2j\pi \right) + 2\varphi_T\left(2j\pi - \pi - \frac{1}{T} \right)\right)\bigg)\end{multlined}
    \end{align}
    Processing the terms separately, we have
    \begin{equation}
        \frac{1}{2}\varphi_T\left(\frac{1}{T}\right)   - \varphi_T\left(\frac{2}{T}\right) = \frac{1}{2}(e^{-\frac{1}{2}} - e^{-2}) \ge 0.167
    \end{equation}
    and if $T\ge\frac{4}{\pi}$,
    \begin{align}
        &\sum_{j\in\mathbb{N}}2\varphi_T\left(2j\pi \right) + 2\varphi_T\left(2j\pi - \pi - \frac{1}{T} \right)\nonumber\\
        &\le\sum_{j\in\mathbb{N}}2\varphi_T\left(2j\pi \right) + 2\varphi_T\left(2j\pi - \pi - \frac{\pi}{4} \right)\nonumber\\
        &\begin{multlined}[t][0.85\linewidth]\le2e^{-2\pi^2T^2} + 2e^{-\frac{9}{32}\pi^2T^2}\\
        + \int_{j=1}^\infty 2\varphi_T\left(2j\pi\right) + \int_{j=3/8}^\infty 2\varphi_T\left(2j\pi\right)
        \end{multlined}\nonumber\\
        &\begin{multlined}[t][0.85\linewidth]\le2e^{-2\pi^2T^2} + 2e^{-\frac{9}{32}\pi^2T^2}\\
        + \frac{1}{2\pi^2 T^2} e^{-2\pi^2T^2} + \frac{4}{3\pi^2 T^2} e^{-\frac{9}{32}\pi^2T^2}
        \end{multlined}\nonumber\\
        &\le\frac{33}{32}e^{-32} + \frac{25}{12}e^{-\frac{9}{2}}<0.024
    \end{align}
    thus in total,
    \begin{equation}
        \mathcal{E}_2(\theta) \ge \frac{0.167-0.024}{1.0008} > 0.14.
    \end{equation}
    
    Thus, we can write 
    \begin{equation}
        \min_{\theta\ge\frac{1}{T}} \widetilde{\mathcal{L}}(\theta) \ge \widetilde{\mathcal{L}}(\lambda) + \min_{\theta\ge\frac{1}{T}}\mathcal{E}_1(\theta) - \sigma\delta  -2\delta^2 +0.14.
    \end{equation}
    Therefore 
    \begin{equation}
        \min_{\theta\ge\frac{1}{T}} \widetilde{\mathcal{L}}(\theta) \ge \max_{\theta<\frac{1}{4T}} \widetilde{\mathcal{L}}(\theta) 
    \end{equation}
    is true under the condition where
    \begin{multline}
        \widetilde{\mathcal{L}}(\lambda) + \min_{\theta\ge\frac{1}{T}}\mathcal{E}_1(\theta) - \sigma\delta  -2\delta^2 +0.14 \\
        \ge \widetilde{\mathcal{L}}(\lambda) + \max_{\theta<\frac{1}{4T}}\mathcal{E}_1(\theta) + \frac{\sigma\delta}{4}  +2\delta^2 + 5\times 10^{-11}
    \end{multline}
    Given that $\lambda <\frac{1}{4T}$, $\min_{\theta\ge\frac{1}{T}}\mathcal{E}_1(\theta) > \max_{\theta<\frac{1}{4T}}\mathcal{E}_1(\theta)$ and the above can be satisfied with a small enough $\delta$, or equivalently, sufficient samples $N$ at a polylogarithmic scale.

    By symmetry, the upper threshold of $\mathcal{A}$ can be found to be the matching $\frac{\pi}{2} - \frac{1}{T}$. Lastly, converting this to amplitude terms, given $\zeta = \sin^2(\frac{1}{4T})$, we find that
    \begin{equation}
        \sin^2\left(\frac{1}{T}\right) \le 16 \sin^2\left(\frac{1}{4T}\right) = 16\zeta.
    \end{equation}
\end{proof}

\subsection{Full-range extension}
\label{appProofPlus}

Lastly, we discuss the proof on how we can extend the algorithm to estimate the full amplitude range by utilizing knowledge of the boundary regions. Similar ideas have been explored for amplitude estimation by \citet{simon2024amplified}, but the method used here is even simpler and more straightforward.

\plus*

\begin{proof}
    Based on \cref{corCerti}, GLSAE rejects the algorithm output in two scenarios, if $\widetilde a < 16 \zeta$ or $\widetilde a > 1 - 16 \zeta$. In this case, the cost of Monte Carlo sampling to estimate $a$ can be found by the Chebyshev inequality as follows:
    \begin{equation}
        \Pr(|\widetilde a - a | \le\epsilon) \le \frac{{\rm Var}(\tilde a)}{\epsilon^2} = \frac{a(1-a)}{N\epsilon^2}
    \end{equation}
    Note that if $a\le 16\zeta$ or $a\ge 1-16\zeta$, then we have
    \begin{equation}
        a(1-a) \in\mathcal{O}(\zeta) \subseteq \mathcal{O}(\varepsilon^{2-2\beta}).
    \end{equation}
    Thus, with failure probability $0.1$, 
    we can find $|\widetilde a - a | \le\epsilon$ with
    \begin{equation}
        N \in \mathcal{O}\left(\frac{a(1-a)}{\epsilon^2}\right) \subseteq \mathcal{O}(\epsilon^{-2\beta})
    \end{equation}
    samples. This can be boosted to a higher arbitrary probability via further median boosting.

    Lastly, to complete the proof, we note that in cases of acceptance in GLSAE, we have, by Lipschitz continuity of cosines,
    \begin{equation}
        \lvert\widetilde{a} - a| \le |\theta-\lambda| \le \epsilon
    \end{equation}

\end{proof}

\section{Additional numerical results}
\label{appNum}

\begin{figure*}
	\centering
    \includegraphics[width=0.95\linewidth]{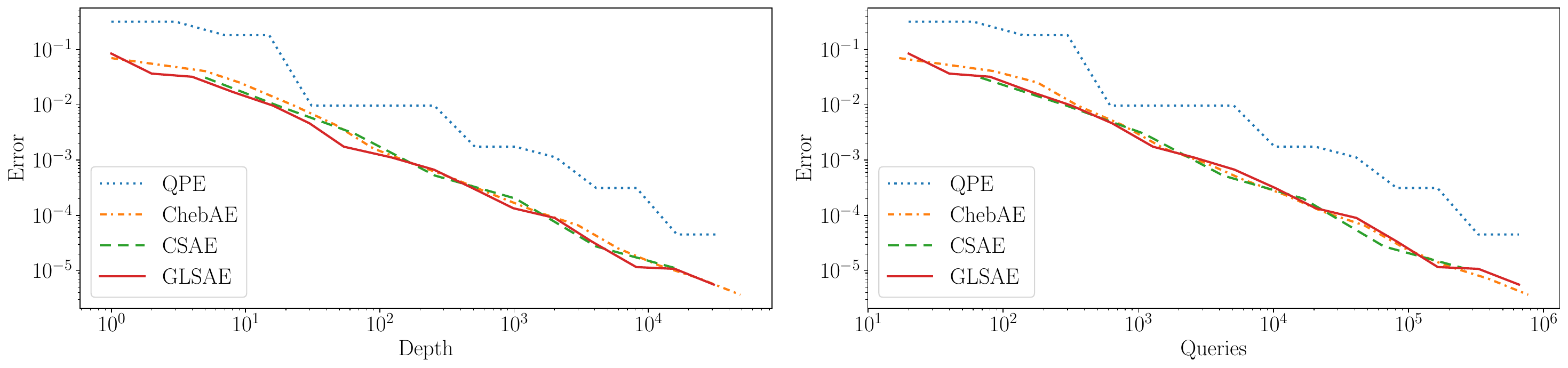}
	\caption{
		\emph{Comparison of the depth and query complexity of Heisenberg-limited amplitude estimation algorithms.} We benchmark GLSAE against the ``textbook'' QPE-based algorithm~\citep{brassard2002quantum}, ChebAE from quantum signal processing~\citep{rall2023amplitude}, and CSAE from ESPRIT~\citep{labib2024quantum}. We note that our methods perform similarly to the prior state-of-the-art, such as ChebAE and CSAE.}
    \label{figHeisenberg}
\end{figure*}

\begin{figure*}
	\centering
    \includegraphics[width=0.95\linewidth]{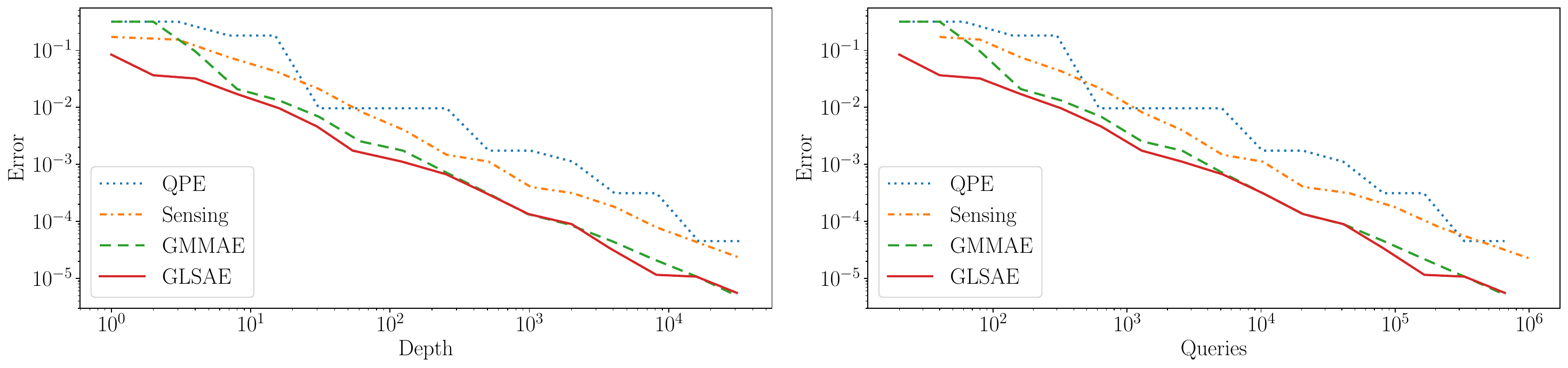}
	\caption{
		\emph{Comparison of the depth and query complexity to additional algorithms.} We benchmark an additional two algorithms, one that uses compressed sensing for the classical post-processing step, and a variant of our algorithm, GMMAE, that maximizes the magnitude function. We show that our highlighted methods in the main text perform better than these variant methods.}
    \label{figOther}
\end{figure*}

\begin{figure*}
	\centering
    \includegraphics[width=0.95\linewidth]{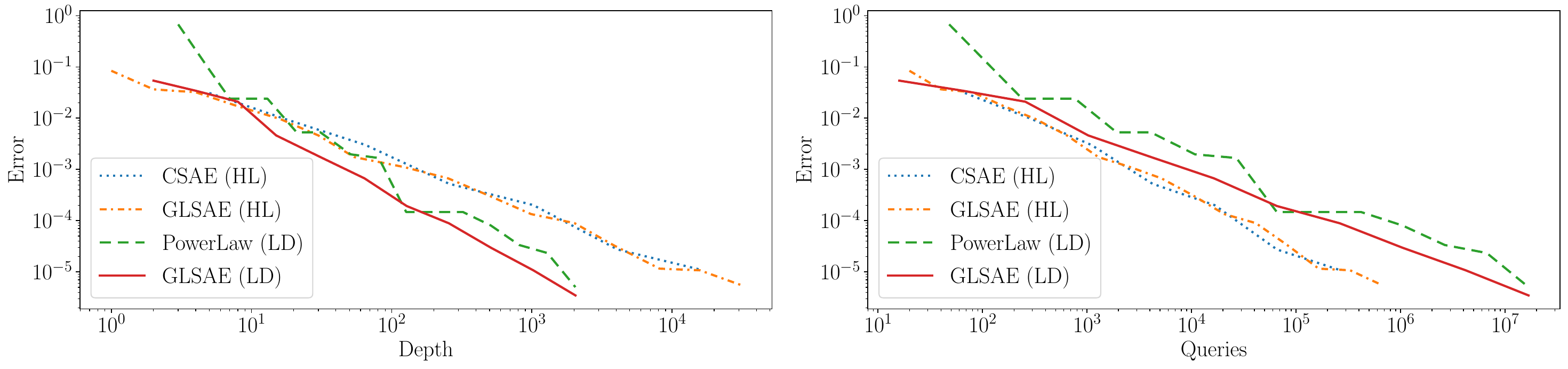}
	\caption{
		\emph{Comparison of the depth and query complexity of amplitude estimation in different regimes.} We set the circuit depth and circuit sample complexity of the low-depth circuits equivalent to uniform time sampling where $\mathcal{D} \in \mathcal{O}(\epsilon^{-2/3})$ and $\mathcal{Q} \in \mathcal{O}(\epsilon^{-4/3})$. Note that the trade-off for low-depth algorithms comes at the cost of higher total query complexity.}
    \label{figDepths}
\end{figure*}

\begin{figure*}
	\centering
    \includegraphics[width=0.95\linewidth]{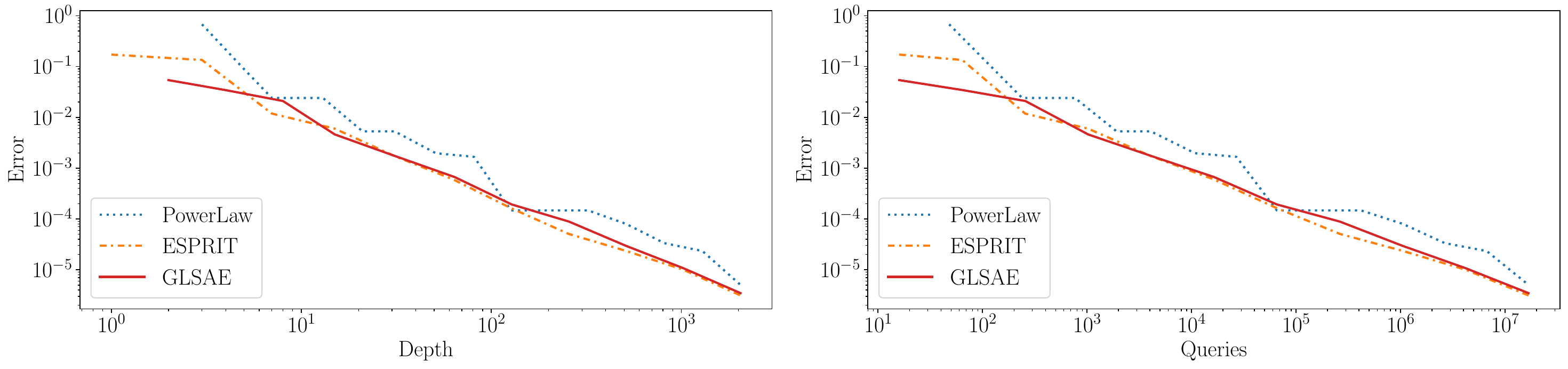}
	\caption{
		\emph{Comparison of the depth and query complexity of low-depth amplitude estimation.} We again set the circuit depth and circuit sample complexity to be equivalent to uniform time sampling. We see that our algorithm outperforms prior art such as Power Law AE~\citep{giurgicatiron2022low}, and matches ESPRIT, which is classically efficient for signal processing under noise.}
    \label{figUniform}
\end{figure*}
We provide more in-depth details for our numerical results in the main text, including full depth and query complexity results, as well as alternate and extended implementations of ideas from statistical phase estimation not featured in the main text.

\subsection{Heisenberg-limited regime}
In \cref{figHeisenberg}, we show that our algorithm performs similarly to the prior state-of-the-art, such as ChebAE and CSAE, in both depth and query complexities. Apart from our algorithm and state-of-the-art amplitude estimation algorithms, we also adapt compressed sensing~\citep{candes2006near-optimal, candes2006robust} for post-processing. Further, we also present a variant algorithm, Gaussian Maximum Magnitude Amplification Estimation. Relevant numerical results are shown in \cref{figOther}. We detail the implementation of our methods as follows:

\paragraph*{Compressed sensing} We adapt the use of compressed sensing in phase estimation from \citet{yi2024quantum} to perform amplitude estimation. As our signals are a single cosine, we employ the use of a discrete Fourier transform matrix, but just the discrete cosine transform matrix. In addition, we employ the use of LASSO~\citep{tibshirani1996regression} for minimization as opposed to convex optimization with $\ell_1$ inequality constraints, as given the high noise in the samples, the inequality constraints require a hyperparameter to control the tolerance of noise, and it might be simpler to penalize the $\ell_1$ norm as part of the loss minimization rather than a strict constraint. From \cref{figOther}, we see that while compressed sensing has asymptotically similar performances (judging from the slope), it requires a higher number of both queries and depth compared to ChebAE, CSAE, and our algorithm. 

\paragraph*{GMMAE} For this algorithm, instead of taking the minimum of the loss function in GLSAE, we maximize the magnitude function
\begin{equation}
\widetilde{\mathcal{M}}'(\theta) = \frac{1}{N} \sum_{m \sim \widetilde{p}_T} Z_m\cos(2\theta m).
\end{equation}
We term this variant algorithm the Gaussian Maximum Magnitude Amplitude Estimation (GMMAE). From \cref{figOther}, we see that the algorithm only operates when the circuit is sufficiently deep, with low-depth circuits having a significantly higher error rate compared to GLSAE. In addition to the indiscernibility of the two Gaussian peaks, as with GLSAE, the magnitude function here has no guarantees that the function would be maximized at $\theta=\lambda$, whereas GLSAE has this guarantee from the nature of least square estimators. This results in a merging of the aforementioned Gaussian peaks and offsets the actual original peaks, distorting the maximum found from its original position and inducing greater errors in a wider range. 

\subsection{Low-depth regime}
In \cref{figDepths}, we show that our algorithm performs better compared to Power Law AE~\citep{giurgicatiron2022low} and exhibits the query--depth tradeoff when compared to Heisenberg-limited amplitude estimation. We quickly note that we do not implement the QoPrime AE algorithm due to the complexity of the post-processing algorithm, as well as due to prior numerical results that indicate an inferiority to Power Law AE in practice~\citep{giurgicatiron2022low,giurgicatiron2022lowexp}. We also omit here the low-depth approaches from \citet{rall2023amplitude} and \citet{vu2025lowdepth}, whose algorithms require a non-Chebyshev polynomial transformation for the purposes of obtaining unbiased estimates, and which would require additional quantum resources (additional phase gates and an additional control qubit) for the implementation of QSP that would be difficult to quantify when compared to the other Grover-type algorithms implemented in this paper.

Given that we select $M, N \in \mathcal{O}(\epsilon^{-2/3})$, the complexity matches that of obtaining the evolution times from a uniform grid. Under this setting, we can apply direction of arrival algorithms from classical signal processing directly, and hence we apply the ESPRIT algorithm as a benchmark, which we detail as follows:

\paragraph*{Off-the-shelf ESPRIT} As our signals are a single cosine, we can adapt \citet{derevianko2022esprit}'s modified ESPRIT algorithm tailored to cosine signal extraction. Compared to \citet{labib2024quantum}'s ESPRIT algorithm for amplitude estimation, standard ESPRIT does not construct virtual arrays nor does it require log-likelihood minimization for sign finding, but does require a uniform sample of the signal array. Our results in \cref{figUniform} show that our algorithm has comparable results to ESPRIT, which is known to be an efficient classical signal processing algorithm even in a high noise regime~\citep{ding2024esprit}.
\end{document}